\documentclass[%
superscriptaddress,
twocolumn,
 nofootinbib,
 amsmath,
 amssymb,
 aps,
 prx,
10pt
]{revtex4-2}

\usepackage[utf8]{inputenc}
\usepackage{graphicx}%
\usepackage{dcolumn}%
\usepackage{bm}%
\usepackage[dvipsnames]{xcolor}
\usepackage[colorlinks=true,linkcolor=teal,citecolor=Maroon,urlcolor=Maroon]{hyperref}
\usepackage{physics}
\usepackage{svg}
\usepackage{siunitx}
\usepackage{float}
\usepackage{tikz}
\usetikzlibrary{shapes.symbols, fit}
\usepackage{mathbbol}
\usepackage{enumitem}
\usepackage{eqparbox}
\usepackage{booktabs}
\usepackage{nicematrix}
\usepackage{lucas}
\usepackage{csquotes}
\usepackage{subcaption}
\usepackage{yquant}
\usepackage{rotating}
\usepackage{mleftright}
\usepackage[binary-units=true, detect-all=true]{siunitx}
\usepackage{threeparttable}
\yquantset{register/default name=}

\usepackage{stmaryrd}
\usepackage{csvsimple}
\usepackage{multirow}

\graphicspath{figures}

\usepackage[nameinlink]{cleveref}
\crefname{prop}{Proposition}{Proposition}
\crefname{exampleCounter}{example}{examples}
\crefname{ex}{example}{examples}
\crefname{lem}{Lemma}{Lemmas}
\crefname{alg}{Algorithm}{Algorithms}
\newcommand{\zl}{Z_L}

\usepackage{tabularx,environ}
\makeatletter
\newcommand{\problemtitle}[1]{\gdef\@problemtitle{#1}}%
\newcommand{\probleminput}[1]{\gdef\@probleminput{#1}}%
\newcommand{\problemquestion}[1]{\gdef\@problemquestion{#1}}%
\NewEnviron{problem}{
  \problemtitle{}\probleminput{}\problemquestion{}%
  \BODY%
  \par\addvspace{.5\baselineskip}
  \noindent
  \begin{tabularx}{.5\textwidth}{@{\hspace{\parindent}} l X c}
    \multicolumn{2}{@{\hspace{\parindent}}l}{\@problemtitle} \\
    \textbf{Input:} & \@probleminput \\
    \textbf{Problem:} & \@problemquestion%
  \end{tabularx}
  \par\addvspace{.5\baselineskip}
}
\makeatother

\begin{document}

\title{Automated Synthesis of \\ Fault-Tolerant State Preparation Circuits for Quantum Error Correction Codes}

\author{Tom Peham}
\email{tom.peham@tum.de}
\affiliation{Chair for Design Automation, Technical University of Munich, Germany}

\author{Ludwig Schmid}
\email{ludwig.s.schmid@tum.de}
\affiliation{Chair for Design Automation, Technical University of Munich, Germany}

\author{Lucas Berent}
\email{lucas.berent@tum.de}
\affiliation{Chair for Design Automation, Technical University of Munich, Germany}

\author{Markus Müller}
\email{markus.mueller@fz-juelich.de}
\affiliation{Institute for Quantum Information, RWTH Aachen University, D-52056 Aachen, Germany }
\affiliation{Peter Grünberg Institute, Theoretical Nanoelectronics, Forschungszentrum Jülich, D-52425 Jülich, Germany}

\author{Robert Wille}
\email{robert.wille@tum.de}
\affiliation{Technical University of Munich}
\affiliation{Software Competence Center Hagenberg, Austria}
\begin{abstract}
A central ingredient in fault-tolerant quantum algorithms is the initialization of a logical state for a given quantum error-correcting code from a set of noisy qubits. 
A scheme that has demonstrated promising results for small code instances that are realizable on currently available hardware composes a non-fault-tolerant state preparation circuit with a verification circuit that checks for spreading errors.
Known circuit constructions of this scheme are mostly obtained manually, and no algorithmic techniques for constructing depth- or gate-optimal circuits exist.
As a consequence, the current state-of-the-art exploits this scheme only for specific code instances and mostly for the special case of distance $d=3$ codes only.
In this work, we propose an automated %
approach for synthesizing fault-tolerant state preparation circuits for arbitrary CSS codes.
We utilize methods based on satisfiability solving (SAT) to construct fault-tolerant state preparation circuits consisting of depth- and gate-optimal preparation and verification circuits. 
We also provide heuristics that can synthesize fault-tolerant state preparation circuits for code instances where no optimal solution can be obtained in an adequate time.
Moreover, we give a general construction for non-deterministic state preparation circuits for codes beyond distance 3.
Numerical evaluations using $d=3$, $d=5$ and $d=7$ codes confirm that the generated circuits exhibit the desired scaling of the logical error rates.
The resulting methods are publicly available as part of the \emph{Munich Quantum Toolkit}~(MQT) at \url{https://github.com/cda-tum/mqt-qecc}.
Such methods are an important step in providing fault-tolerant circuit constructions that can aid in near-term demonstrations of fault-tolerant quantum computing.
\end{abstract}
\maketitle

\section{Introduction}
Quantum bits and quantum operations suffer from unavoidable decoherence and noise. 
Therefore, quantum error-correction mechanisms need to be used to ensure that a quantum algorithm can be executed in a fault-tolerant manner to control the accumulation of errors.
Error-correcting codes leverage redundancy by encoding quantum information in logical qubits formed by entangled states of noisy, physical qubits.
Consequently, it is important that universal operations can be carried out on the encoded information fault-tolerantly
and that errors during the execution can be detected and corrected.
A crucial step in any fault-tolerant circuit is the preparation of an encoded logical state.

To ensure that the logical state is initialized correctly -- without containing too many errors -- the preparation itself needs to be fault-tolerant. 
Intuitively, a requirement for fault-tolerant circuits is that errors on single qubits cannot uncontrollably ``spread'' to multiple qubits, for instance, through CNOT gates.
For \emph{Calderbank-Shor-Steane} (CSS) codes~\cite{calderbankGoodQuantumErrorcorrecting1996,steaneMultipleParticleInterference1996,steaneErrorCorrectingCodes1996}, one way of initializing an encoded state fault-tolerantly is to prepare all physical qubits in $\ket{0}$ or $\ket{+}$ and perform one round of error correction---projecting the product state onto the logical code space~\cite{nielsenQuantumComputationQuantum2010}.
For large code instances, this overhead might be manageable, especially for QLDPC codes and single-shot codes~\cite{breuckmannQuantumLowdensityParitycheck2021, bombinSingleShotFaultTolerantQuantum2015, delfosseSingleShotFaultTolerantQuantum2022, quintavalleSingleShotErrorCorrection2021, fawziConstantOverheadQuantum2020}.

This general procedure introduces significant overhead in circuit size, especially for small code instances that are relevant for near-to-midterm devices since all stabilizer generators must be measured multiple times in an error correction round.
Consequently, the near-term experimental realization of this scheme is difficult.
\begin{figure*}[t]
  \centering
  \captionsetup{
    labelfont=bf,        %
    justification=raggedright,
}
\includegraphics[width=\textwidth]{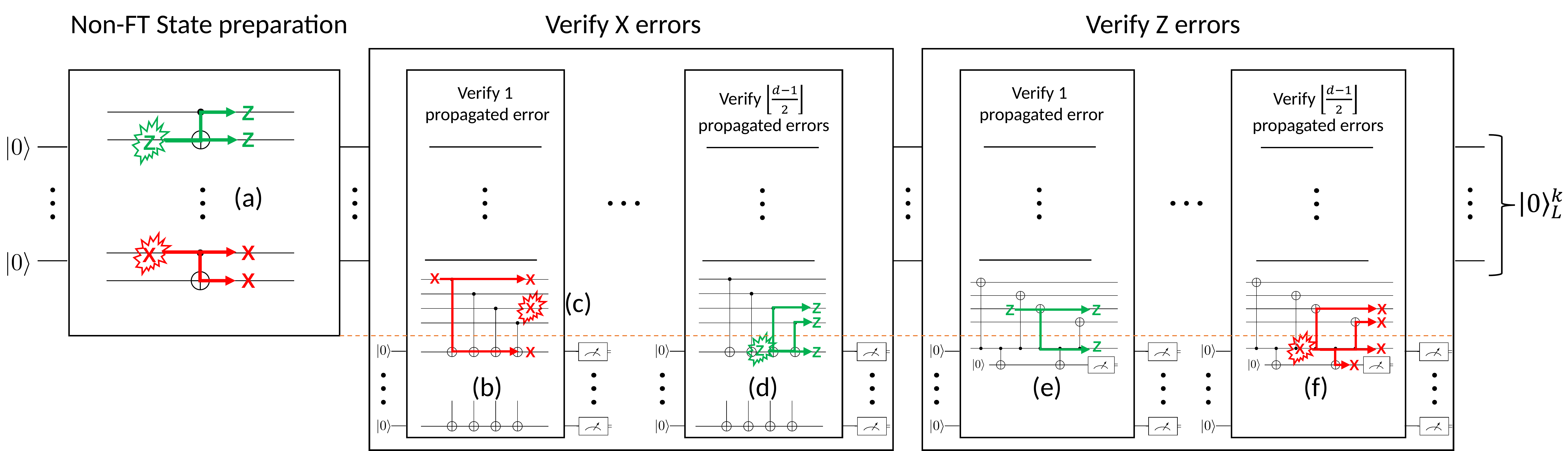}
  \caption{Full non-deterministic fault-tolerant state preparation protocol for CSS codes. 
  A sequence of stabilizer measurements follows a non-fault-tolerant state preparation circuit to check if any errors propagated through the circuit. 
  The state is accepted if no measurement in the verification blocks indicates an error. 
  \textbf{(a)} Errors in the non-fault-tolerant state preparation circuits propagate to higher-weight errors through CNOTs between data qubits. 
  \textbf{(b)} If an error of at most weight $i$ occurred in the state preparation circuit and propagated to a higher-weight error, measurements in the $i$th layer of the verification circuit detect this error. 
  \textbf{(c)} Propagated errors are still detected by later layers of verification, even if further errors occur during a stabilizer measurement. 
  \textbf{(d)} Z errors on the ancilla during the verification of X errors propagate to the data qubits. 
  \textbf{(e)} If a Z error propagated through the non-FT state preparation circuit \emph{or} the Z measurements of the previous verifications, it is later detected by flag fault-tolerant measurements. 
  \textbf{(f)} Propagated X errors on the X measurement ancillas are detected by flag qubit measurements.} 
  \label{fig:full-ft}
\end{figure*}
An alternative that tries to reduce the number of measurements is a non-deterministic protocol, based on post-selection~\cite{gotoMinimizingResourceOverheads2016,bermudezFaulttolerantProtectionNearterm2019,postlerDemonstrationFaulttolerantUniversal2022, chamberlandFaulttolerantMagicState2019}. 
Therein, the initialization is performed in two steps:
\begin{enumerate}
    \item Prepare the logical state using a non-fault-tolerant state preparation circuit.
    \item Conduct specific measurements on the prepared states using a so-called \emph{verification circuit} that detects if an error spreads through the circuit. 
    If at least one of these measurements measures a $-1$ eigenvalue, the state is discarded, and the protocol is restarted.
\end{enumerate} 
While this and related techniques were successfully demonstrated in recent in experiments~\cite{buttFaultTolerantCodeSwitchingProtocols2024, bermudezFaulttolerantProtectionNearterm2019, heussenMeasurementFreeFaultTolerantQuantum2024, postlerDemonstrationFaulttolerantUniversal2022,ryan-andersonRealizationRealtimeFaulttolerant2021, bluvsteinLogicalQuantumProcessor2024, m.p.dasilvaDemonstrationLogicalQubits2024, pogorelovExperimentalFaulttolerantCode2024}, there are clear drawbacks of current approaches: 
First, most of these techniques rely on manual construction, which is tedious and becomes infeasible when scaling to larger codes, where both the number of errors and the number of possible verification measurements increase quickly. 
Moreover, as for the general scheme outlined above, there is, in general, no guarantee that these manually constructed circuits are gate- or depth-optimal.
Finally, these constructions have mostly been explored for specific code instances for small distance 3 (or 2) codes~\cite{buttFaultTolerantCodeSwitchingProtocols2024,postlerDemonstrationFaulttolerantUniversal2022, gotoMinimizingResourceOverheads2016, zenQuantumCircuitDiscovery2024} that constitute a special case. 
Even for cases beyond distance 3, constructions of non-deterministic circuits were derived manually for specific states~\cite{chamberlandFaulttolerantMagicState2019} or codes~\cite{goswamiFaulttolerantPreparationQuantum2023, goswamiFactorybasedFaulttolerantPreparation2024}.
In general, constructing a verification circuit that ensures fault-tolerance of the overall scheme is highly non-trivial and depends on the code (and the logical state to prepare).
Thus, the manual construction in existing works is not generally applicable.
Moreover, even though post-selection schemes do not scale well to large distances due to the exponentially decreasing acceptance rate, optimal circuit constructions might still be viable for moderately sized codes on near-term hardware. Furthermore, the improvements to circuit size for small codes are amplified when concatenating codes, as the same state preparation circuit can be used for higher levels of concatenation where a single CNOT translates to a transversal CNOT between all data qubits of two code states. In particular, with the continuously improving error rates on current devices, experimental demonstrations of higher-distance codes might be possible with a verification-based scheme. 
To this end, the construction needs to be generalized beyond distance 3, and algorithmic methods that can be applied to a broad range of codes must be developed.

In this work, we address these issues and propose automated methods for the synthesis of fault-tolerant state preparation circuits for \emph{arbitrary} CSS codes. 
Our main contributions are as follows:
\begin{itemize}
    \item We separate the synthesis problems for a state preparation and verification circuit and tackle them individually. 
    We frame both problems as binary optimization problems and discuss their computational complexity to provide a motivation for the techniques we use.
    \item Based on that, we utilize satisfiability solving (SAT) techniques, which, through the use of highly optimized solving algorithms, have been shown to be successful in dealing with small instances of optimization problems -- both in the classical domain~\cite{biereSATbasedModelChecking2018, brandVerificationLargeSynthesized1993, eggersglussImprovedSATbasedATPG2013, gebregiorgisTestPatternGeneration2019, kaufmannVerifyingLargeMultipliers2019, willeSMTbasedStimuliGeneration2009,khomenkoLogicSynthesisAsynchronous2004} and for quantum computing~\cite{willeMappingQuantumCircuits2019, tanOptimalLayoutSynthesis2020, tanSATScalpelLattice2024, shuttyDecodingMergedColorSurface2022, pehamDepthOptimalSynthesisClifford2023}. 
    Thereby, we obtain exact solutions and, consequently, gate- or depth-optimal state preparation~%
    circuits and verification circuits that are gate-optimal with respect to a given state preparation circuit. 
    \item 
    To address the scalability issues of SAT techniques, we also propose \emph{heuristic} solutions for synthesizing state preparation and verification circuits. 
    While these do not guarantee optimal solutions, this will allow us to obtain solutions for larger problem instances 
    \item To use the proposed methods for synthesizing logical basis states for CSS codes with distances greater than 3, we propose a generalized non-deterministic fault-tolerant state preparation scheme that guarantees fault-tolerance by composing multiple verification circuits.
    The resulting scheme is summarized in \Cref{fig:full-ft}. 
\end{itemize}
In short, our contributions can be viewed as providing (optimal) circuits for every sub-circuit depicted in~\Cref{fig:full-ft}. 
To investigate the error suppression of the circuits obtained with the proposed methods, we synthesize non-deterministic fault-tolerant state preparation circuits for various distance 3 CSS codes and provide numerical simulations of the logical error rate under a circuit-level noise model using Stim~\cite{gidneyStimFastStabilizer2021}. 
Comparison with the state-of-the-art reinforcement learning-based synthesis method from Ref.~\cite{zenQuantumCircuitDiscovery2024} shows that our methods are competitive---producing equally good or better circuits than the reinforcement learning agent. 
Moreover, numerical evaluations show that the proposed scheme does indeed give circuits that exhibit the desired logical error rate scaling for codes beyond distance 3. 
Finally, all proposed methods are made publicly available in the form of open-source software as part of the \emph{Munich Quantum Toolkit}~(MQT~\cite{willeMQTHandbookSummary2024}) at \hyperlink{https://github.com/cda-tum/mqt-qecc}{https://github.com/cda-tum/mqt-qecc}.

This manuscript is structured as follows. 
In~\Cref{sec:background}, we review notation, give fundamental notions from quantum coding, and define the notion of a fault-tolerant state preparation protocol.
~\Cref{sec:problem-motivation} describes the problem of synthesizing non-deterministic fault-tolerant state preparation circuits and positions our work in the context of existing approaches. 
In~\Cref{sec:state-prep}, we describe our optimal and heuristic synthesis methods for non-fault-tolerant state preparation for CSS codes. 
How to verify such circuits and make them fault-tolerant is then described in~\Cref{sec:ver-circ}, where we also discuss the computational complexity of this problem. 
Based on these building blocks, we then detail our generalized scheme for non-deterministic fault-tolerant state preparation in~\Cref{sec:fully-fault-tolerant}.
~\Cref{sec:eval} shows the resulting circuit for our synthesis method for various CSS codes and numerically verifies that our proposed scheme prepares logical basis states of CSS codes in a fault-tolerant manner. Finally,~\Cref{sec:conclusion} concludes this paper.

\section{Background}\label{sec:background}

In this section, we review formal notation and discuss fundamental concepts that are needed throughout the work.
Throughout we use $[i] = \{1,2,\dots, i\}\subseteq \mathbb{Z}$. We make use of the \emph{Iverson Bracket}, which maps Boolean expressions to integers $\llbracket P \rrbracket =
\begin{dcases*}
  1 & if P is true \\
  0 & otherwise
\end{dcases*}$.

\subsection{Quantum Codes}\label{sec:prelim-codes}
An $\llbracket n,k,d \rrbracket$-stabilizer code~\cite{gottesmanStabilizerCodesQuantum1997} is defined by a stabilizer group $S$, which is an abelian subgroup of the $n$-qubit Pauli group $\mathcal{P}_n$.
We require $-I \notin S$ s.t.\ the code is non-trivial.
The set of codewords, i.e., the $k$-dimensional codespace, is defined as the simultaneous $+1$ eigenspace of all elements in $S$.
The generators of the stabilizer group $S = \langle g_1, g_2, \dots g_m\rangle$ are called \emph{checks}.
They are used to infer whether a state is a codeword by measuring all $m$ checks (using additional ancillas). 
The $m$ measurement outcomes can be modeled as vector $s \in \mathbb{F}_2^m$ called the \emph{syndrome}.
If all syndrome bits are $0$ (indicating a $+1$ measurement of the eigenvalue), the state is a codeword.
Otherwise, the syndrome indicates which checks have failed and thus gives information about the approximate location of the error that occurred.
The problem of finding a suitable recovery operation given the syndrome is called \emph{decoding} and is a purely classical task.

A stabilizer code on $n$ physical qubits with $m$ checks encodes $k = n-m$ \emph{logical} qubits. 
The \emph{distance} $d$ of the code is the minimum weight of an operator transforming one codeword into another.

\begin{figure}[t]
  \centering
  \includegraphics[scale=0.25]{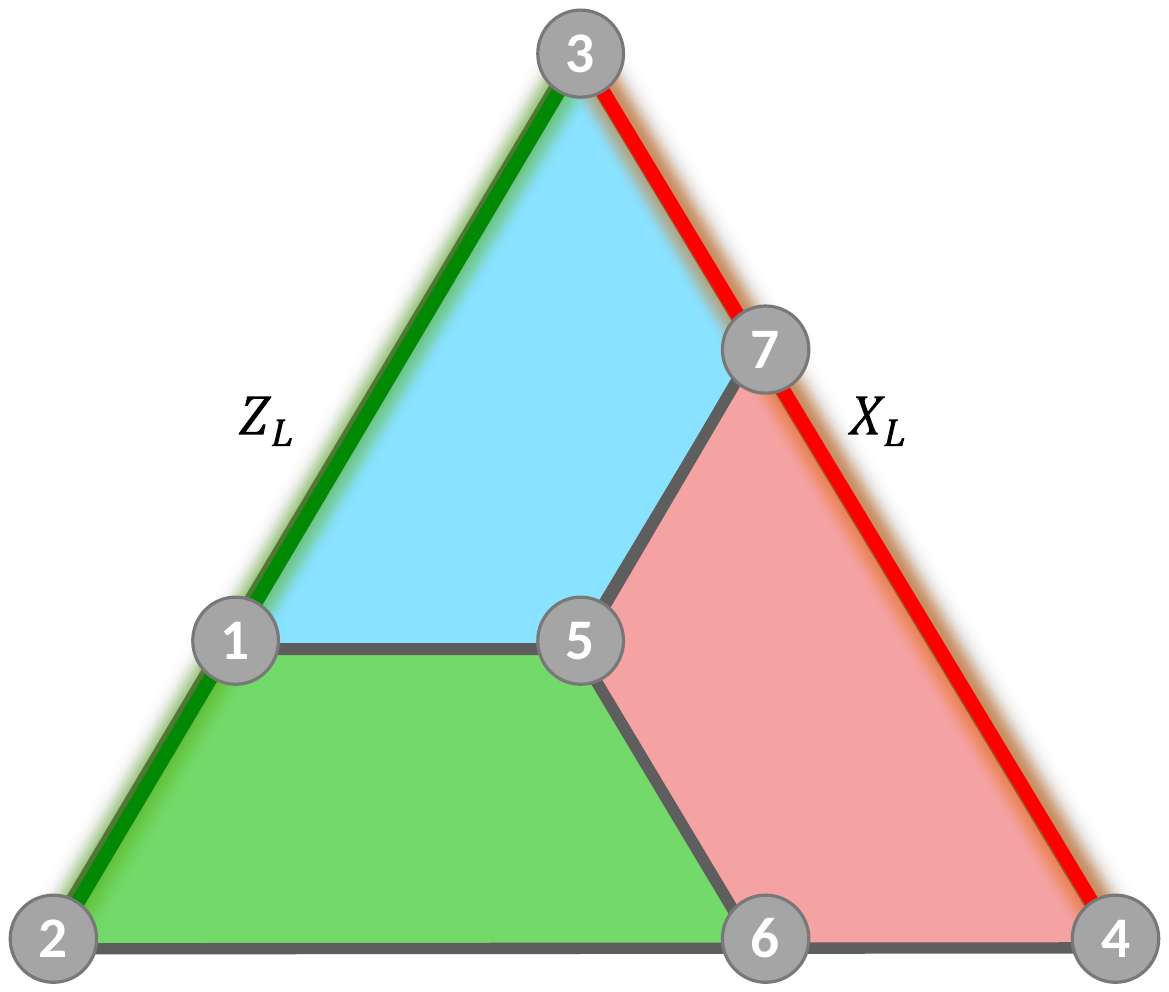}
  \caption{Steane code. The qubits are placed on vertices and X- and Z-stabilizer checks are associated with faces, acting on adjacent vertices.}
  \label{fig:steane-code-ex}
\end{figure}
      
A particularly important class of stabilizer codes are Calderbank-Shor-Steane (CSS) codes~\cite{calderbankGoodQuantumErrorcorrecting1996,steaneMultipleParticleInterference1996}, whose stabilizer generators $g_i$ are either all X, or all Z. 
The corresponding sets are denoted $S_X$ and $S_Z$, respectively.
Note that $X$ errors only anticommute with $Z$ stabilizers and $Z$ errors only anticommute with $X$ stabilizers. Therefore, $X$ and $Z$ type errors can be decoded independently.

\begin{example}
    The distance $3$ color code, also called \emph{Steane code}, is the smallest instance of a family of 2D planar color codes~\cite{bombinIntroductionTopologicalQuantum2013}. 
    It can be visualized as depicted in~\Cref{fig:steane-code-ex}, where qubits are associated with vertices and stabilizer generators with faces. 
    Hence, it is a $\llbracket 7,1,3 \rrbracket$ CSS code defined by the stabilizer generators $S_X = \{X_1X_2X_5X_6, X_1X_3X_5X_7, X_4X_5X_6X_7\}$ and $S_Z = \{Z_1Z_2Z_5Z_6, Z_1Z_3Z_5Z_7, Z_4Z_5Z_6Z_7\}$. Logical operators are associated with the sides of the triangle. For example $X_L = X_3X_4X_7$ and $Z_1Z_2Z_3$ would be minimal-weight logical $X$- and $Z$-operators (indicated by the red and green lines in~\Cref{fig:steane-code-ex}). Due to the symmetry of the $X$- and $Z$-stabilizers, $X_L = X_1X_2X_3$ and $Z_L = Z_3Z_4Z_7$ are also valid logical operators.
\end{example}

The stabilizer generators in $S_X, S_Z$ can be mapped to binary vectors that indicate the support of the respective Pauli operators, e.g., $XXII \in S_X, XXII\mapsto (1,1,0,0)\in \mathbb{F}_2^n$.
Hence, a CSS code can equivalently be defined as two binary matrices called \emph{parity-check matrices}, $H_X\in \mathbb{F}_2^{r_X\times n}, H_Z\in \mathbb{F}_2^{r_Z\times n}$ whose rows are the binary vectors corresponding to the stabilizer generators. 
Since the $X$ and $Z$ type generators have to commute, the matrices have to fulfill the orthogonality condition $H_X H_Z^T = 0$. 

\begin{example}
    The parity-check matrices of the Steane code are identical:
    \begin{equation}
        H_X = H_Z = \left( 
            \begin{matrix}
                1&  1 &  & & 1 &1 & \\
                 1 &  & 1 & & 1 & &1   \\
                 &   & & 1 & 1 & 1 &1
            \end{matrix}
        \right).
    \end{equation}
\end{example}

\subsection{Noise Model and Fault-Tolerance}\label{sec:prelim-noise-model}

We use a standard, depolarizing circuit-level noise model parameterized by the noise strength $p$, as discussed in more detail in~\Cref{sec:eval}.

The weight of an error $E$ is defined as \mbox{$\mathrm{wt}(E) = \min_{E'\in E \cdot S} |\mathrm{supp}(E')|$} i.e., as the minimal support of a stabilizer reduced error with respect to the stabilizer group $S$ of the code. 
We say that an error \emph{propagates} through a two-qubit gate if an error on one qubit before the gate leads to a two-qubit error after the two-qubit gate. 
We call an error in a circuit $\mathcal{C}$ a \emph{propagated error} if it propagates through any gate in $\mathcal{C}$. 

There are slightly different meanings to the term \emph{fault-tolerance} in the literature~\cite{aliferisQuantumAccuracyThreshold2005, gottesmanIntroductionQuantumError2009, chamberlandFlagFaulttolerantError2018}.
Intuitively, the definition we follow means that the uncontrolled spreading and accumulation of errors is avoided.
More formally, we adhere to the following definition for a fault-tolerant state preparation protocol.
\begin{definition}[\textbf{Fault-tolerant state preparation}~\cite{paetznickFaulttolerantAncillaPreparation2013,derksDesigningFaulttolerantCircuits2024}]\label{def:ft}
  A state preparation circuit is \emph{strictly fault-tolerant} if for all $t \leq \floor{\frac{d}{2}}$, any error of probability order $t$ propagates to an error of weight at most $t$.
\end{definition}

\section{Problem Definition \& Motivation}\label{sec:problem-motivation}

In this section, we illustrate the problem of fault-tolerant state preparation with verification circuits and motivate the need for generalized automated approaches.

\begin{figure}[t]
  \centering
  \captionsetup{
    labelfont=bf,        %
    justification=raggedright,
}
\includegraphics[width=.5\textwidth]{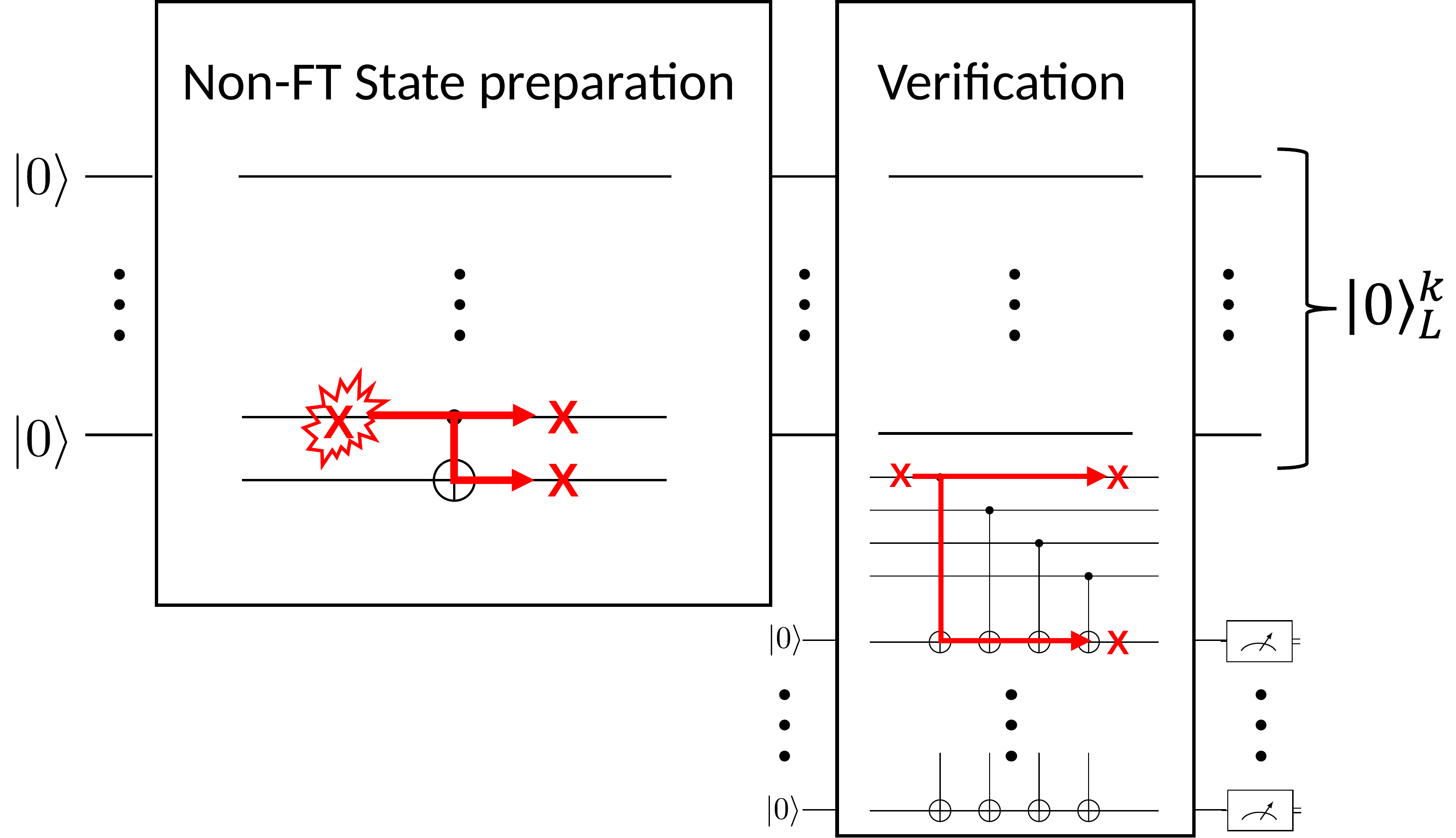}
  \caption{Fault-tolerant Pauli-eigenstate preparation for stabilizer codes. A state is prepared with a non-fault-tolerant encoding circuit, and the state is post-selected on a set of stabilizer measurements called the verification circuit. If one of these measurements flags, the state is discarded, and the process is restarted.}
  \label{fig:scheme-overview}
\end{figure}
\subsection{Considered Problem}
Before starting any fault-tolerant quantum computation, the logical qubits have to be initialized in some logical state---usually a Pauli eigenstate of the quantum error correcting code used to encode the quantum information. 
These states are easily described for an $\llbracket n, 1, d \rrbracket$ CSS code. 
The logical $\ket{0}_L$ state of the code is defined by the stabilizers of the code and the logical $\zl$ operator.
Similarly, the logical $\ket{+}$ state is defined by the stabilizers of the code and the logical $X_L$ operator. 
The Pauli eigenstates of a general $\llbracket n, k, d\rrbracket$ CSS code are defined analogously for corresponding pairs of logical operators $X^i_L, Z^i_L,i\in [k]$. 
For the remainder of this work, we will focus on the preparation of the logical $\ket{0}_L^{\otimes k}$ state, but the methods transfer directly to preparing the $\ket{+}^{\otimes k}$ state by swapping the roles of $X$ and $Z$- stabilizers and operators and reversing the direction of every CNOT in the state preparation circuits.

\begin{figure}[t]
  \centering
    \captionsetup{
    labelfont=bf,        %
    justification=raggedright,
}
  \begin{tikzpicture}
    \begin{yquant}
      qubit {$q_1: \ket{0}$} q[1];
      qubit {$q_2:\ket{+}$} q[+1];
      qubit {$q_3:\ket{+}$} q[+1];
      qubit {$q_4:\ket{+}$} q[+1];
      qubit {$q_5: \ket{0}$} q[+1];
      qubit {$q_6: \ket{0}$} q[+1];
      qubit {$q_7: \ket{0}$} q[+1];

      cnot q[0] | q[1];
      cnot q[5] | q[3];
      cnot q[6] | q[2];
      cnot q[4] | q[1];
      cnot q[0] | q[2];
      cnot q[6] | q[3];
      cnot q[5] | q[1];
      cnot q[4] | q[6];

      [after=q]
      qubit {$\ket{0}$} a[1];
      cnot a[0]|q[0];
      cnot a[0]|q[5];
      cnot a[0]|q[6];
      measure a;
    \end{yquant}
  \end{tikzpicture}
  \caption{Non-deterministic state preparation for the $\ket{0}_L$ of the Steane code from Ref.~\cite{gotoMinimizingResourceOverheads2016}.}
  \label{fig:steane-non-deterministic}
\end{figure}

A simple protocol for preparing Pauli eigenstates is to initialize all physical qubits in $\ket{0}$ and to project the product state onto the code space by measuring the stabilizers of the code. 
The logical $Z^i_L$ operators are already stabilizers of the initial $\ket{0}^{\otimes n}$ and will still be a stabilizer of the projected state since the operators $Z^i_L$ commute with the stabilizers of the code. 
To ensure fault-tolerance in the presence of noisy syndrome measurements, the stabilizer measurements have to be repeated a number of times that is typically proportional to the code distance~\cite{dennisTopologicalQuantumMemory2002,bombinIntroductionTopologicalQuantum2013}.
For Shor-type QEC, the number of required repetitions is in $O(d^2)$~\cite{tansuwannontAdaptiveSyndromeMeasurements2023} for instance.
Additional CNOT overhead is incurred for some measurement schemes when using a flag-fault-tolerance scheme to measure the stabilizers~\cite{chamberlandFlagFaulttolerantError2018}.
For the Steane code, this procedure requires $6\cdot (4+2)\cdot 3 = 108$ CNOTs in the worst case, i.e., the $6$ weight-$4$ stabilizer generators of the Steane code have to be measured $3$ times---the distance of the Steane code---and each measurement requires another $2$ CNOTs for the flags.
For experimental realizations in near-term quantum computers, the error rate of these CNOTs will completely overshadow any error suppression theoretically guaranteed by the Steane code.

For this reason, recent experimental demonstrations of fault-tolerance have used a so-called \emph{non-deterministic} scheme based on post-selection to produce high-fidelity logical states~\cite{buttFaultTolerantCodeSwitchingProtocols2024, bermudezFaulttolerantProtectionNearterm2019, heussenMeasurementFreeFaultTolerantQuantum2024, postlerDemonstrationFaulttolerantUniversal2022,ryan-andersonRealizationRealtimeFaulttolerant2021, bluvsteinLogicalQuantumProcessor2024, m.p.dasilvaDemonstrationLogicalQubits2024, pogorelovExperimentalFaulttolerantCode2024}. 
This approach aims to prepare the logical state in a non-fault-tolerant manner and measure a set of stabilizers (including logical operators) that detect any uncorrectable error caused by a correctable error propagating through the state preparation circuit. 
If one of these measurements is triggered, the state is discarded, and the whole process is restarted. 
The measurement circuit is referred to as the \emph{verification circuit}.
The protocol is sketched in~\Cref{fig:scheme-overview}. 
The verification circuit depends on the specific state preparation circuit, as the errors that the verification measurements should detect are determined by the CNOT gate pattern in the state preparation circuit.

\begin{example}\label{ex:steane-verification}
  For the  $\ket{0}_L$ state of the Steane code, the optimal non-deterministic preparation scheme uses 8 CNOTs for the preparation circuit and a single weight-three measurement of a logical $Z_L$ operator~\cite{gotoMinimizingResourceOverheads2016} (cf.~\Cref{fig:steane-non-deterministic}) in the verification circuit. 
  The only problematic stabilizer-inequivalent errors that can occur in the circuit are $X_5X_7$ and $X_2X_6$. 
  These anti-commutes with the measured $Z_L=Z_1Z_6Z_7$ operator are thus detected in this verification measurement.
\end{example}

Ideally, the verification circuit only indicates an error if a propagated error occurs. 
Unfortunately, this is impossible in general. For example, any single qubit error on the measurement ancilla will trigger the verification measurements. The \emph{acceptance rate} $r_A$ of a non-deterministic preparation scheme is the probability that a state will be accepted, i.e., that no measurement in the verification circuit is triggered. This probability decreases exponentially with the number of measurements even when ignoring errors in the state preparation circuit itself: $p_A\leq (1-p)^m$ where $m$ is the number of measurements.

This work aims to formulate a general algorithm that synthesizes a (near-)optimal circuit for each of these subcircuits for any $\llbracket n, k, d\rrbracket$ CSS code. 

\subsection{Related Work}
\label{sec:related-work}

Recent work has tackled state preparation and verification circuit synthesis using reinforcement learning (RL)~\cite{zenQuantumCircuitDiscovery2024}. 
The biggest difference between our method and the RL approach is that we only focus on CSS codes, while the RL framework works for arbitrary stabilizer codes. Furthermore, we ignore qubit connectivity constraints in this work while this aspect is accounted for in Ref.~\cite{zenQuantumCircuitDiscovery2024}. 
Lastly, the RL framework also provides a method for tackling the state preparation and verification circuit synthesis problem \emph{simultaneously}, whereas we treat them separately. 
The downside of the RL-based approach is that it does not give any optimality guarantees. 
Since the code is a parameter in the training of the RL agent, the runtimes for synthesizing a single circuit are high, even for small codes. 

Since CSS code state preparation circuits consist only of CNOT gates~\cite{kissingerPhasefreeZXDiagrams2022}, the state preparation circuit synthesis problem is closely related to the synthesis of linear reversible circuits. 
This is a well-researched problem with known algorithms reaching the theoretical lower bound on circuit size and depth~\cite{jiangOptimalSpaceDepthTradeOff2022,patelOptimalSynthesisLinear2008,maslovDepthOptimizationCZ2022}. 
These algorithms provide systematic ways that are \emph{asymptotically optimal}. These asymptotic analyses hide a significant constant overhead that we can't ignore for the smaller system sizes we are mainly interested in. 
That is why our state preparation circuit synthesis approach is more closely related to greedy CNOT circuit synthesis methods as proposed in Ref.~\cite{timotheegoubaultdebrugiereGaussianEliminationGreedy2021}. 
The most important difference between the considered problems and linear circuit synthesis is that we have more degrees of freedom. 
Although both problems can be framed in terms of matrix elimination over $\mathbb{F}_2$, the matrices in our setting are not necessarily invertible. 
Therefore, many different circuits prepare the same state despite implementing different linear maps.

Another closely related problem is that of general Clifford and stabilizer state preparation circuit synthesis. 
For these, there are also optimal~\cite{bravyi6qubitOptimalClifford2022,schneiderSATEncodingOptimal2023,pehamDepthOptimalSynthesisClifford2023} and asymptotically optimal~\cite{aaronsonImprovedSimulationStabilizer2004,bravyiCliffordCircuitOptimization2021,maslovLinearDepthStabilizer2007,maslovDepthOptimizationCZ2022} methods. Our SAT-based optimal state preparation circuit synthesis method is closely related to Ref.~\cite{pehamDepthOptimalSynthesisClifford2023} and Ref.~\cite{schneiderSATEncodingOptimal2023}. 
The optimal synthesis in Ref.~\cite{pehamDepthOptimalSynthesisClifford2023} is focused only on the depth-optimal synthesis of Clifford circuits and does not provide an encoding to obtain (two-qubit) gate-optimal circuits. 
More importantly, as with the synthesis of linear circuits, there are more degrees of freedom that previous work does not take advantage of when focusing only on state preparation. 
The previous work on optimal stabilizer state preparation of Ref.~\cite{schneiderSATEncodingOptimal2023} does not make use of this freedom either and does not guarantee the synthesis of a globally optimal circuit. 
Optimality is only guaranteed with respect to a specific stabilizer generator set. In the extreme case, Ref.~\cite{schneiderSATEncodingOptimal2023} synthesizes a linear-sized circuit for the $\ket{0}^{\otimes n}$ state if a particularly unfortunate representation of its stabilizers is given as an input. 
For more complex states, it is not clear which generating set for the state's stabilizer group yields the best circuit. 
We address this specifically in~\Cref{sec:state-prep}.

\section{Synthesis of\\State Preparation Circuits}\label{sec:state-prep}

In this section, we present a method for creating logical basis states for any CSS code given the stabilizer generators of the code. 
We focus on creating $\ket{0}^{\otimes k}_L$ state, but the method works exactly the same for the $\ket{+}^{\otimes k}_L$ state by swapping the roles of $X$ and $Z$ stabilizer and changing directions of CNOTs.
For the rest of this section, we consider an $\llbracket n, k, d \rrbracket$ CSS code with $m$ $X$ stabilizer generators.

\begin{figure*}[t!]
  \centering
   \captionsetup{
    labelfont=bf,        %
    justification=raggedright,
}
  \includegraphics[width=\textwidth]{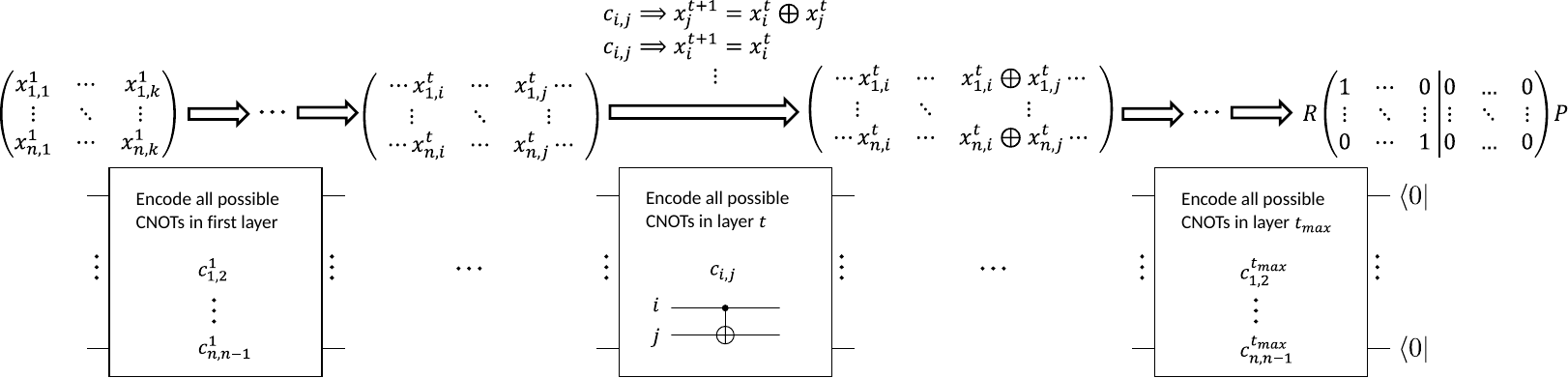}
  \caption{Symbolic encoding for optimal state preparation circuit synthesis. The circuit is synthesized using at most $t_\mathrm{max}$ layers of CNOTs. All possible column additions on the symbolic check matrices are related via variables $c_{i,j}^t$, which encode all possible CNOTs in layer $t$. Starting from the target check matrix $H_X$ these constraints encode all possible valid circuits that prepare the logical $\ket{0}^{\otimes k}_L$ state.}
  \label{fig:sp-encoding}
\end{figure*}

\subsection{Problem Description and Complexity} 
\label{sec:task-descr-compl}
Since the orthogonal complement of the $X$ check matrix $H_X \in \mathbb{F}_2^{m \times n}$ is uniquely defined and coincides with the space spanned by the rows of $H_Z$ and logical operators $\zl^i, i\in [k]$, $\ket{0}^{\otimes k}_L$ is uniquely fixed by the $X$ stabilizers of the code alone.
An important property of logical X or Z basis state preparation circuits of CSS codes is that they consist only of physical qubits initialized in the $X$ or $Z$ basis followed by a sequence of CNOTs. In fact, there is a one-to-one correspondence between such circuits and CSS codes~\cite{kissingerPhasefreeZXDiagrams2022}. 
If we have $m$ qubits initialized in $\ket{+}$ and $n-m$ qubits in $\ket{0}$ then the product state of these qubits has $m$ $X$-stabilizers and $n-m$ $Z$-stabilizers. Applying a CNOT to this state maps $X$-stabilizers to $X$-stabilizers and $Z$-stabilizers to $Z$-stabilizers. Therefore, any CNOT circuit applied to this initial product state does not change the number of $X$- or $Z$-stabilizers. Combined with the uniqueness of the orthogonal complement, this tells us that if the state obtained by applying a CNOT circuit to the initial product state is stabilized by the $X$-stabilizers of the code and no other $X$-stabilizers, it must necessarily also be stabilized by the $Z$-stabilizers and the logical $Z$-operators of the code and is, therefore, the $\ket{0}_L^{\otimes k}$ of the code.
These properties simplify the considered problems significantly because, for state preparation, we only have to consider $X$ stabilizers.

Any intermediate state during preparation is also a logical basis state of some CSS code with respective $X$ checks $H_X^a$, where $H_X^a$ denote the $X$-checks of the partial state after the first $a$ CNOTs.
If the $a+1$st CNOT has control $q_i$ and target $q_j$ then $H_X^{a+1}$ is obtained by adding column of $H_X^a$ $i$ onto column $j$ of $H_X^a$. This comes from the fact that a Pauli $X$ copies through the control of a CNOT.
A CSS state preparation circuit $\mathcal{C}$ on $n$ qubits, therefore, implements an invertible matrix $M_{\mathcal{C}}$ over $\mathbb{F}_2$ converting $H_X^0$ to the final $H_X$.

Recall that the rows of $H_X$ represent stabilizer generators for the $X$ stabilizers of the given code.
The choice of stabilizer generators is not unique in general. 
In particular, we can replace any two generators $\mathcal{X}_i, \mathcal{X}_j$ by equivalent generators $\mathcal{X}_i, \mathcal{X}_i\cdot \mathcal{X}_j$. 
We say that two check matrices $H_X$ and $H_X^\prime$ are equivalent if $H_X = R \cdot H_X^\prime$ for some invertible $\mathbb{F}_2^{m\times m}$ matrix $R$.

The task of CSS state preparation circuit synthesis, then, is to find a CNOT circuit $\mathcal{C}$ such that $H_XM_\mathcal{C} = R \cdot \begin{bmatrix} I & 0 \end{bmatrix}$ or, equivalently $H_X = R \begin{bmatrix} I & 0 \end{bmatrix}M_\mathcal{C}^{-1}$. 
Here, we assumed without loss of generality that the qubits initialized in $\ket{+}$ are the first $m$ qubits in the qubit ordering.
The second formulation of the synthesis problem suggests that instead of constructing the circuit by successively adding CNOTs until we obtain $H_X$, we can start from $H_X$ and apply CNOTs until we obtain the \emph{reduced} check matrix $\begin{bmatrix} A & 0 \end{bmatrix}$ where $A$ is some rank $\mathrm{rk}~(H_X)\times \mathrm{rk}~(H_X)$ matrix.
The state preparation circuit is then obtained by initializing the first $m$ qubits in the $\ket{+}$ state and applying the CNOTs used in the reduction in \emph{reverse} order.

Complexity-wise, generating optimal CNOT circuits is closely related to the minimum circuit size problem, which is unlikely to be proven to be in P or to be NP-complete~\cite{kabanetsCircuitMinimizationProblem2000}. Therefore, it is unlikely to find a polynomial algorithm for state preparation circuit synthesis since it is simply a variant of that problem. 
In the following, we describe how to convert an instance of the state preparation circuit synthesis problem to a SAT formula and obtain gate- and depth-optimal preparation circuits. 
Since this reduction to SAT is not expected to scale due to the intractability of the SAT problem, we also provide a polynomial-time heuristic algorithm for this problem.

\subsection{Optimal Approach}\label{sec:opt-state-prep}

We would like to find an optimal preparation circuit whenever possible. 
Fewer CNOTs naturally lead to lower error rates in any circuit, and shallower circuits lead to faster execution time, which in turn leads to a decrease in idling errors. 
Therefore, we want to synthesize \emph{gate}- and \emph{depth}-optimal circuits.

Since the synthesis problem is formulated over $\mathbb{F}_2$, SAT solvers~\cite{biereHandbookSatisfiability2009} are a natural choice for this task. SAT solvers have previously been employed to synthesize general stabilizer circuits~\cite{pehamDepthOptimalSynthesisClifford2023} and state preparation circuits in particular~\cite{schneiderSATEncodingOptimal2023}. 
However, these approaches provide a solution only for a few qubits. 
We can improve upon these SAT encodings by taking advantage of the fact that we are dealing with CSS codes.

The idea behind the encoding is sketched in \Cref{fig:sp-encoding}. Essentially, the task is to symbolically encode Gaussian elimination over $\mathbb{F}_2$. 
To this end, we first introduce \emph{matrix variables} $\{x_{i,j}^t \mid i \in [m], j \in [n], t \in [t_{\mathrm{max}}+1]\}$, where $t_{\mathrm{max}}$ is the given maximum number of elimination steps. 
The maximum number of elimination steps allows us to control how many CNOTs are applied at each step. 
Depending on the specific optimization task, we can use $t_{\mathrm{max}}$ to limit the total number of CNOTs or the depth of the circuit.

Initially, we assert 

$$\forall i \in [m].~\forall  j \in [n].~ x_{i,j}^1=H_X[i,j],$$

where $H_X[i, j]$ denotes the matrix entry in row $i$ and column $j$ of $H_X$.

For the reduced matrix, we assert the pseudo-boolean constraint

$$\left(\sum_{j=1}^{n} \llbracket \bigvee_{i=1}^{m} x_{i,j}^{t_{\mathrm{max}}+1}\rrbracket\right) = \mathrm{rk}~(H_X).$$

In other words, the matrix must be reduced to a matrix $H_X^{t_\mathrm{max}+1}$ with exactly $\mathrm{rk}~H_X$ columns with non-zero entries. 
This is equivalent to asserting that there exists an invertible matrix $R$ and a permutation matrix $P$ such that $R H_X^{t_\mathrm{max}} = P \begin{bmatrix} I_m & 0 \end{bmatrix}$. 
This constraint guarantees that the choice of stabilizer generators has no impact on the synthesized circuit's size.

Symbolic matrices in consecutive steps are obtained via column operations. 
Either a column is added onto another, or a column does not change. 
Therefore, we introduce two Boolean formulas

\begin{align*}
  \mathrm{add}(i,j, t) &= \forall k \in [m].~ x_{k,j}^{t+1}=x_{k, i}^t \oplus x_{k,j}^t \\
  \mathrm{id}(i, t) &= \forall k \in [m].~ x_{k, i}^{t+1}=x_{k,i}^t.\\
\end{align*}

No matter which of the two metrics we optimize for, obtaining the optimal state preparation circuit is achieved by solving the encoded formula for different values of $t_{\mathrm{max}}$ until two values $t_1, t_2$ are found such that

\begin{itemize}
\item The formula is unsatisfiable for $t_{\mathrm{max}} = t_1$ and
\item The formula is satisfiable for $t_{\mathrm{max}} = t_2$.
\end{itemize}

The optimal solution is then the one obtained from the variable assignment to the satisfiable instance.

In the following, we will give the details of the encoding for synthesizing depth- and gate-optimal circuits. 

\subsubsection*{Encoding for Depth-optimal Circuits}
\label{sec:encod-depth-optim}

For synthesizing depth-optimal circuits, we interpret $t_{\mathrm{max}}$ as the maximal circuit depth. 
We then need only to encode all possible layers of CNOTs at each reduction step.
If the solver finds a satisfying assignment, we can extract the circuit directly from the variable assignment. 
If the solver shows no such assignment exists, we can incrementally vary $t_{\mathrm{max}}$ until a satisfying assignment is found.

We introduce \emph{CNOT variables} \mbox{$\{c_{i,j}^t \mid i \in [n], j \in [n], t \in [t_{\mathrm{max}}+1]\}$}, which encode all possible CNOTs (or column additions) from control qubit (column) $i$ to target qubit (column) $j$. 
These variables encode how the symbolic matrices at steps $t$ and $t+1$ are related:

\begin{align}\label{eq:add-constraint}
  \begin{split}
  \forall t \in [t_{\mathrm{max}}].~\forall i \in [n].
  \forall  j \in [n]\setminus \{i\}.\\ c_{i,j}^t \implies \mathrm{add}(i, j, t).    
  \end{split}
  \end{align}

  Additionally, any qubit can be involved in at most one CNOT at one timestep:

$$\forall t \in [t_{\mathrm{max}}].~\forall i \in [n].~(\sum_{j=1}^{n}\llbracket c_{i,j}^t \rrbracket + \sum_{j=1}^{n}\llbracket c_{j,i}^t \rrbracket) \leq 1.$$

Lastly, if a qubit is not the target of any CNOT, the corresponding column does not change:

\begin{align*}
  &\forall t \in [t_{\mathrm{max}}].~\forall i \in [n]. \bigl(\neg \bigwedge_{j=1}^{n}c_{j,i}^t\bigr) \implies \mathrm{id}(i, t).
\end{align*}

\subsubsection*{Encoding for Gate-optimal Circuits}\label{sec:encod-gate-optim}

For the synthesis of gate-optimal circuits, instead of encoding entire layers of CNOTs, symbolic matrices in consecutive elimination steps are related by only a single column addition, ensuring that at every time step, only one CNOT can be performed. The maximum number of time steps, then, corresponds to the maximum number of CNOTs allowed in the circuit construction. 
This can be achieved by augmenting the depth-optimal encoding with further constraints prohibiting more than one CNOT occurring at every step as done in~\cite{schneiderSATEncodingOptimal2023}. 
Alternatively, these additional constraints can be avoided by encoding all possible control and target combinations for each CNOT instead of every possible CNOT. 
Using binary encoding for the control and target qubit indices, we require only $2\floor{\log_2n}$ instead of $n(n-1)$ variables per timestep.

Thus, we introduce \emph{bit-vector} variables $\{\mathrm{ctrl}^t \mid 1 \leq t \leq  t_{\mathrm{max}}\}$ and $\{\mathrm{tar}^t \mid 1 \leq t \leq  t_{\mathrm{max}}\}$ of size $\ceil{\log_2n}$. Since bit-vector logic is not supported by SAT solvers, the following formulation requires using a solver for \emph{Satisfiability Module Theories} (SMT) that supports bit-vector logic. We can then encode \emph{CNOT constraints} relating the control and target variables:

$$c_{i,j}^t \Leftrightarrow \left( \mathrm{ctrl}^t = i \land \mathrm{tar}^t=j \right).$$

We can symbolically encode the column addition constraints using \cref{eq:add-constraint}.

Encoding the fact that a column does not change if the corresponding qubit is not the target of a CNOT is simpler in this setting:

\begin{align*}
\forall t \in [t_{\mathrm{max}}].~\forall i \in [n]. \mathrm{tar}^t \neq i \implies \mathrm{id}(i, t).  
\end{align*}

Of course, a qubit cannot be both the target and control of a CNOT:

$$\forall t \in [t_{\mathrm{max}}]. \mathrm{ctrl}^t\neq \mathrm{tar}^t$$

Since the bit-vectors might take on values outside the allowed qubit range, we impose the following constraint:

$$\forall t \in [t_{\mathrm{max}}]. (1 \leq \mathrm{ctrl}^t \leq n) \land (1 \leq \mathrm{tar}^t \leq n).$$

With these constraints encoded, we can try to obtain a solution using publicly available SMT solvers like Z3~\cite{leonardodemouraZ3EfficientSMT2008}.

\subsection{Heuristic Approach}\label{sec:heur-state-prep}

CNOT and stabilizer circuit synthesis have been well researched in the past, and various asymptotically optimal algorithms exist~\cite{jiangOptimalSpaceDepthTradeOff2022,maslovLinearDepthStabilizer2007,patelOptimalSynthesisLinear2008}. 
For small circuits, the constant overhead in these constructions leads to unacceptably large circuits.

In Ref.~\cite{timotheegoubaultdebrugiereGaussianEliminationGreedy2021,schaefferCostMinimizationApproach2014}, the authors propose a greedy path-finding-based method for synthesizing linear circuits. 
For an input $\mathbb{F}_2$ matrix $H$ they define the cost functions

\begin{align*}
  h_{\text{sum}}(H)&=\sum_{i,j}H[i,j]\\
  h_{\text{prod}}(H)&=\sum_{i}\log(\sum_j H[i,j]),
\end{align*}

where addition is taken over $\R$ rather than $\F_2$.

The first cost function simply counts the number of non-zero entries in the matrix.
The idea behind the second cost function is that it weights columns that are almost completely zeroed higher.

We can use these cost functions to synthesize state preparation circuits by starting with the target check matrix $H_X$ and greedily choosing the CNOT that minimizes the cost of the matrix resulting from applying the CNOT, i.e., the corresponding column-addition.
This is repeated until the cost can no longer be minimized, either because the reduced check matrix corresponds to a product state of $\ket{0}$ and $\ket{+}$ states or because the search is stuck in a local minimum.
In Ref.~\cite{timotheegoubaultdebrugiereGaussianEliminationGreedy2021}, an iteration limit is set before stopping the calculation to avoid getting stuck in a local minimum indefinitely. 
In our case, we can use the fact that the stabilizers form a group by changing the basis of the row space to escape local minima when the synthesis algorithm is stuck. 
To this end, we can also greedily pick the row additions that remove the most ones~\footnote{In principle, these reductions can be tried between every column operation to reduce the number of 1s in the matrix even faster. 
We have not observed a definite benefit of doing this as it sometimes actually leads to worse solutions.}. 
If even this approach fails, we can always fall back to doing Gaussian elimination on the rows of the matrix.
We make further use of this degree of freedom by preprocessing the matrix.
The synthesized circuit will probably be quite deep if a matrix column has substantially more non-zero entries than other columns. 
This can be prevented by applying row operations that decrease the maximum column weight without decreasing the row weights.

Refs.~\cite{timotheegoubaultdebrugiereGaussianEliminationGreedy2021,schaefferCostMinimizationApproach2014} also propose to use the combined cost of $H$ and $H^{-1}$ to guide the search.
Since the matrices, in our case, are not invertible, this is not possible.
Note that reducing the check matrix using such heuristics is a similar idea to the approach proposed in Ref.~\cite{paetznickFaulttolerantAncillaPreparation2013} where the \enquote{overlap} of stabilizers is used to reduce the CNOT count in state preparation circuits.

This greedy search strategy seems more promising for the smaller codes we focus on. 
Using the described cost functions for a heuristic search will optimize for gates, but we also want to optimize for depth as much as possible. 
A simple heuristic for this is to construct the circuit layer by layer, always greedily choosing CNOTs that minimize the cost function until all qubits are involved in a gate or until no possible CNOT would decrease cost anymore.

\section{Synthesis of Verification Circuits}\label{sec:ver-circ}

State preparation circuits like the ones we synthesize above are generally not fault-tolerant. 
The key idea to ensure fault tolerance is to append a verification circuit to the preparation circuit that indicates whether an error has propagated in the circuit and post-selecting on the verification measurement outcomes. 
More specifically, a verification circuit is the circuit implementation of a specific set of stabilizer or logical operator measurements of the state such that any propagated error anticommutes with one of the verification measurements. An error anticommutes with a stabilizer measurement if their support overlaps on an \emph{odd} number of qubits. Thus, the propagated errors of the state preparation circuit impose a parity constraint on the measurements used in the verification circuit.

In the following, we will describe the problem of synthesizing verification circuits in detail. In particular, we describe the challenges that arise when generalizing verification circuit synthesis to codes with distances larger than 3. We show how the verification circuit synthesis problem can be broken down into synthesizing a sequence of verification circuits that verify a specific set of errors---significantly reducing the complexity of the problem.

To synthesize each subcircuit, we show how to encode the synthesis problem into a SAT formula. Iterative queries to a SAT solver can then again be used to find the gate- or ancilla-optimal sub-verification circuit. As with the optimal synthesis of state preparation circuits, there will be limits to scaling such an optimal synthesis method. We, therefore, also propose a scalable, greedy synthesis scheme.

\subsection{Problem Description and Complexity}
\label{sec:task-descr-compl-ver}

A simple way to synthesize a verification circuit is to measure all stabilizer generators, which might include many superfluous measurements.
Note again that since we prepare logical basis states, the respective logical operators are also stabilizers of the state.

Define the $i$-\emph{fault set} $\mathcal{E}_i(\mathcal{C})$ of a circuit $\mathcal{C}$ as the set of propagated errors $E$ with $\mathrm{wt}(E)>i$ that are caused by the propagation of an error of probability order $i$. 
Computing this set itself takes $O(\#\mathrm{CNOTs}^{i})$ time and is therefore fixed-parameter tractable in the number of CNOTs. 
Since we will mainly look at small codes with a small number of CNOTs, computing this set is straightforward.

Let us continue to focus on the logical $\ket{0}^{\otimes k}_L$ state. Since we are looking at CSS codes, $X$ and $Z$ errors can be checked separately. Let us focus on verifying $X$ errors for now, i.e., the verification circuit construction described in the following will only be strictly fault-tolerant against $X$-type errors.
How this construction can be extended to be strictly fault-tolerant for general errors will be shown in \Cref{sec:fully-fault-tolerant}.
To check for $X$ errors, any element of $S_Z$ (including logical operators) can be measured in general. 
Searching for a minimally-sized set of measurements seems infeasible considering the size of this group. This is stated by the following theorem.

\begin{theorem}\label{thm:np-grp}
  Consider the following decision problem: Given an X fault set $\mathcal{E}$, Z stabilizer generators $Z_1, \cdots, Z_m$ and an integer $k$, is there a set of group elements $Z_{\mathrm{ver},1}, \cdots, Z_{\mathrm{ver},t}$ such that $\sum_{i=1}^{t} \mathrm{wt}(Z_{\mathrm{ver},i}) \leq k$ and for all $E \in \mathcal{E}$ there is a $Z_{\mathrm{ver},i}$ such that $E\cdot Z_{\mathrm{ver},i} = - Z_{\mathrm{ver},i}\cdot E$? This problem is NP-complete.
\end{theorem}

\begin{proof}
  The problem is clearly in NP, since given a certificate for an instance it is straightforward to verify that it is a valid solution. Furthermore, this is just a more general version of the Coset leader problem, which asks for a minimum-weight representative of the coset of some linear code. Computing such minimal-weight representatives is NP-complete~\cite{berlekampInherentIntractabilityCertain1978} and is essentially a decoding problem. 
\end{proof}

Consider an $\llbracket n, k, d \rrbracket$ code which protects against up to $t = \floor{\frac{d-1}{2}}$ errors and a state preparation circuit $\mathcal{C}$ for the $\ket{0}^{\otimes k}_L$ of this code. Generating a verification circuit for all $X$ errors $\bigcup_{i=1}^{\floor{\frac{d-1}{2}}} \mathcal{E}_i(\mathcal{C})$ does not ensure fault-tolerance in and of itself because errors can also occur in the verification circuit itself. We distinguish the following cases.

\begin{enumerate}
\item $t$ errors occur in the non-fault-tolerant state preparation circuit and propagate to an \emph{uncorrectable} error $E_\mathcal{C} \in \mathcal{E}_t$ at the output.
\item $i < t$ errors occur in the non-fault-tolerant state preparation circuit and propagate to a \emph{correctable} error $E_\mathcal{C} \in \mathcal{E}_i, E_\mathcal{C} \notin \mathcal{E}_t$ at the output.
\item $i < t$ errors occur in the non-fault-tolerant state preparation circuit and propagate to an \emph{uncorrectable} error $E_\mathcal{C} \in \mathcal{E}_i \cap \mathcal{E}_t$ at the output.
\end{enumerate}

\begin{figure}[t]
    \centering
       \captionsetup{
    labelfont=bf,        %
    justification=centerlast,
}
    \centering
    \includegraphics[width=.4\textwidth]{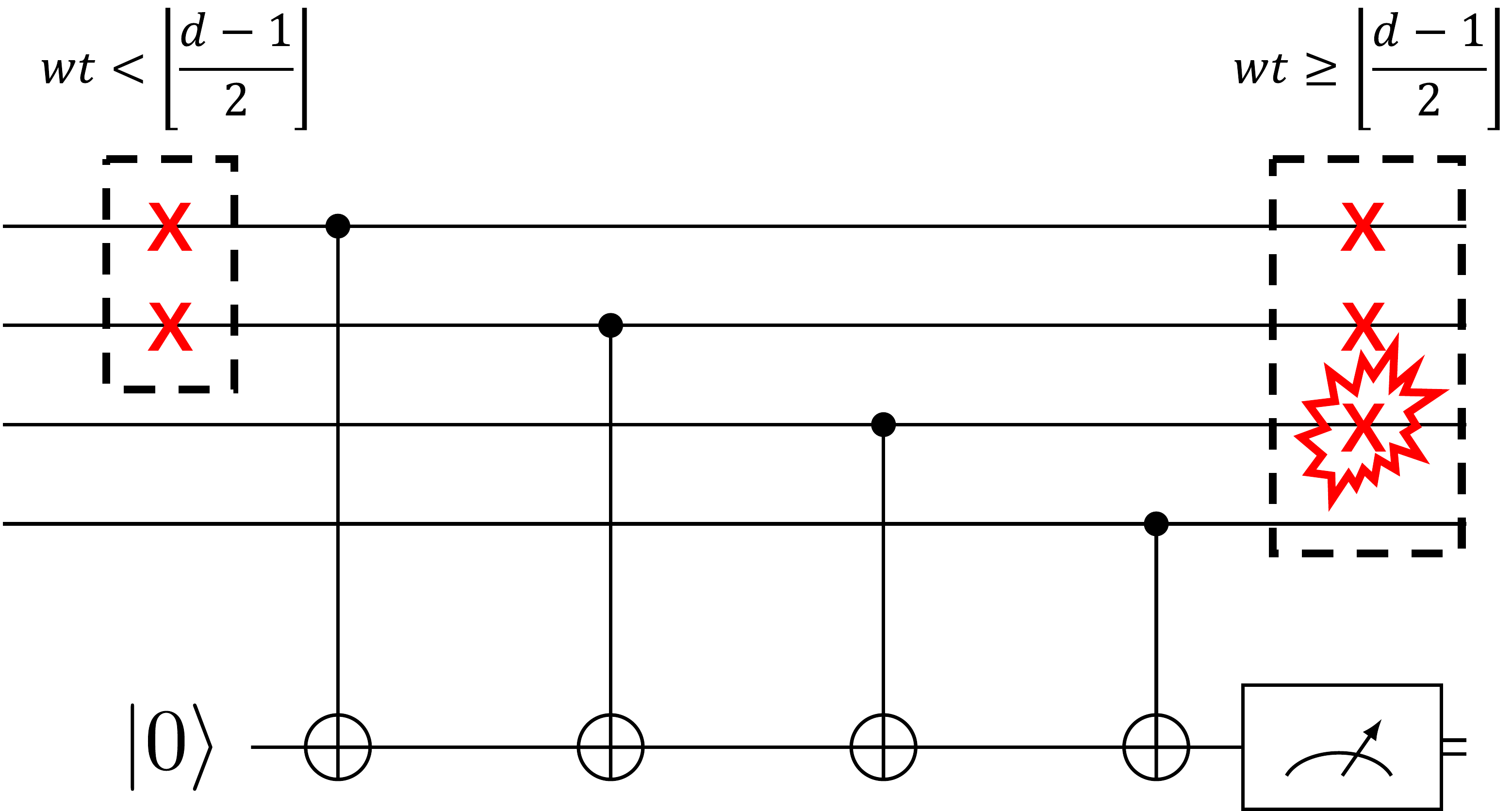}
  \caption{A correctable error is (correctly) not detected by the verification circuit, but an error after the stabilizer measurement turns the error into an uncorrectable one. If this happens after the last measurement, which would detect this error, it goes by undetected. }\label{fig:error-case-1}
\end{figure}

The first case will always be detected by the verification circuit, assuming that at most $t$ errors occur in the entire state preparation circuit.
In the second case, if another error $E_\mathrm{Ver}$ of weight $\mathrm{wt}(E_\mathrm{Ver})=t-i$ occurs in the verification circuit, then it could be that \mbox{$E_\mathcal{C} \cdot E_\mathrm{Ver} \notin \mathcal{E}_i$}. 
If this error is not in the fault set $ \mathcal{E}_t$, then this error is equivalent to some non-propagated error. 
If it is in the fault set, however, it should be detected by some measurement in the verification circuit.
But it could be that $E_\mathrm{Ver}$ happens after every measurement that would detect this combined error. Then, the state will not be discarded, and the fault-tolerance condition of~\Cref{def:ft} is violated(see also~\Cref{fig:error-case-1}).

In the third case, the verification circuit would detect the error, but if no more than $t-i$ measurements detect this error, then all these measurements could be flipped with at most $t-i$ independent errors. We say that the error is \emph{masked}. 
The verification circuit would then indicate no error, and~\Cref{def:ft} is violated again (see~\Cref{fig:error-case-2}).

For distance $3$ codes for which fault-tolerant verification circuit-based state preparation schemes exist, this is not much of an issue since in cases $2$ and $3$, the state preparation circuit would be error-free, i.e., $i=0$.
Any weight $w$ $X$ error in the verification circuit can lead to, at most, a weight $w$ $X$ error on the data qubits since the CNOTs of the measurement circuit only propagate $X$ errors onto ancilla qubits, so case $2$ is unproblematic. Even an error on a measurement ancilla will only trigger one of the verification measurements, leading to the state being discarded. For distance $3$ codes, the verification circuit can, therefore, be treated as error-free without penalty.

Greater care must be put into synthesizing the verification circuit for larger distances. 
The case distinction above already suggests the solution. 
The verification circuit needs to detect not only uncorrectable propagated errors but correctable ones, too, as these might lead to a logical error together with additional errors in the rest of the circuit. 
Let $\mathrm{Ver}_i$ be a verification circuit for $\mathcal{E}_i$.
We achieve fault-tolerance against one type of error ($X$ or $Z$) with the verification circuit 
\begin{equation}
  \label{eq:ver}
\mathrm{Ver} = \mathrm{Ver}_t  \circ \mathrm{Ver}_{t-1} \circ \cdots \circ \mathrm{Ver}_1,  
\end{equation}
where $\circ$ denotes circuit composition.

As before, if any measurement in this circuit indicates an error, the state is discarded, and the process is started anew. 
This ensures fault tolerance by the following reasoning:

Intuitively, if a layer in the verification circuit fails to detect a propagated error due to errors in the verification circuit itself, the resulting error will either be detected by a later verification or the error is equivalent to an unpropagated error.

More precisely, let $E$ be a propagated error with $\mathrm{wt}(E)>i$ caused by $i < t$ errors in the state preparation circuit.
By construction, $E$ will be detected by $\mathrm{Ver}_i$.
Observe that a single-qubit $X$ error on a qubit in the $\ket{+}$ state at the beginning of the circuit will propagate to a stabilizer operator of the code and act trivially on the prepared state.
As a result, if $E \in \calE_i(\calC)$ then $$E \in \bigcap_{j=i}^{\mathrm{wt}(E)}\calE_j(\calC),$$ meaning that if $E$ is the result of $i$ propagated errors in the state preparation circuit, it could also have been the result of $i+1, \cdots, \mathrm{wt}(E)$ errors (the additions errors would have been trivial). 
This means that $E$ will also be detected by each $\mathrm{Ver}_{i+1}, \cdots, \mathrm{Ver}_{\mathrm{wt}(E)}$.

The only way $E$ can pass verification undetected, then, is that further errors in the verification circuit mask it.
Any such error is either a single-qubit error on the data qubits or a single-qubit error on an ancilla qubit.
A two-qubit $X$ error on a CNOT is equivalent to a single-qubit $X$ error on the control before the CNOT, so we do not have to give special consideration to this case.

From the perspective of $\mathrm{Ver}_{i+1}$, it is immaterial whether an error on the data qubits occurs in some previous verification layer or at the end of the state preparation circuit.
This is because it does not matter if the combined error is detected by previous verification layers or not, as it will definitely be detected by $\mathrm{Ver}_{i+1}$.
Generally for an error $E'$ resulting from $u$ independent single-qubit errors on data qubits such that $\mathrm{wt}(E\cdot E') > i+u$ it is the case that $E\cdot E' \in \calE_{i+u}(\calC)$ and thus $E \cdot E'$ will be detected by $\mathrm{Ver}_{i+u}$.

Similarly, an error on an ancilla qubit is indistinguishable from a trivial error on the data qubits when considering the verification circuit $\mathrm{Ver}_{i+1}$.
Again, if an error $E'$ resulting from $v$ independent errors on ancilla qubits occurs and $\mathrm{wt}(E) > i+v$, then $E \in \calE_{i+v}(\calC)$ and thus $E\cdot E'$ will be detected by $\mathrm{Ver}_{i+v}$.

Combining the two cases, for an error $E'$ resulting from $u$ independent single-qubit errors occur on data qubits and $v$ single-qubit errors on ancilla qubits in the verification circuit such that $i+u+v \leq t$ and $\mathrm{wt}(E \cdot E') > i+u+v$ then $E\cdot E' \in \calE_{i+u+v}(\calC)$ and thus $E\cdot E'$ will be detected by $\mathrm{Ver}_{i+u+v}$.
Therefore, the circuit $\mathrm{Ver}$ is fault-tolerant.

\begin{figure}[t]
    \centering
    \captionsetup{
    labelfont=bf,        %
    justification=centerlast,
}
    \includegraphics[width=.4\textwidth]{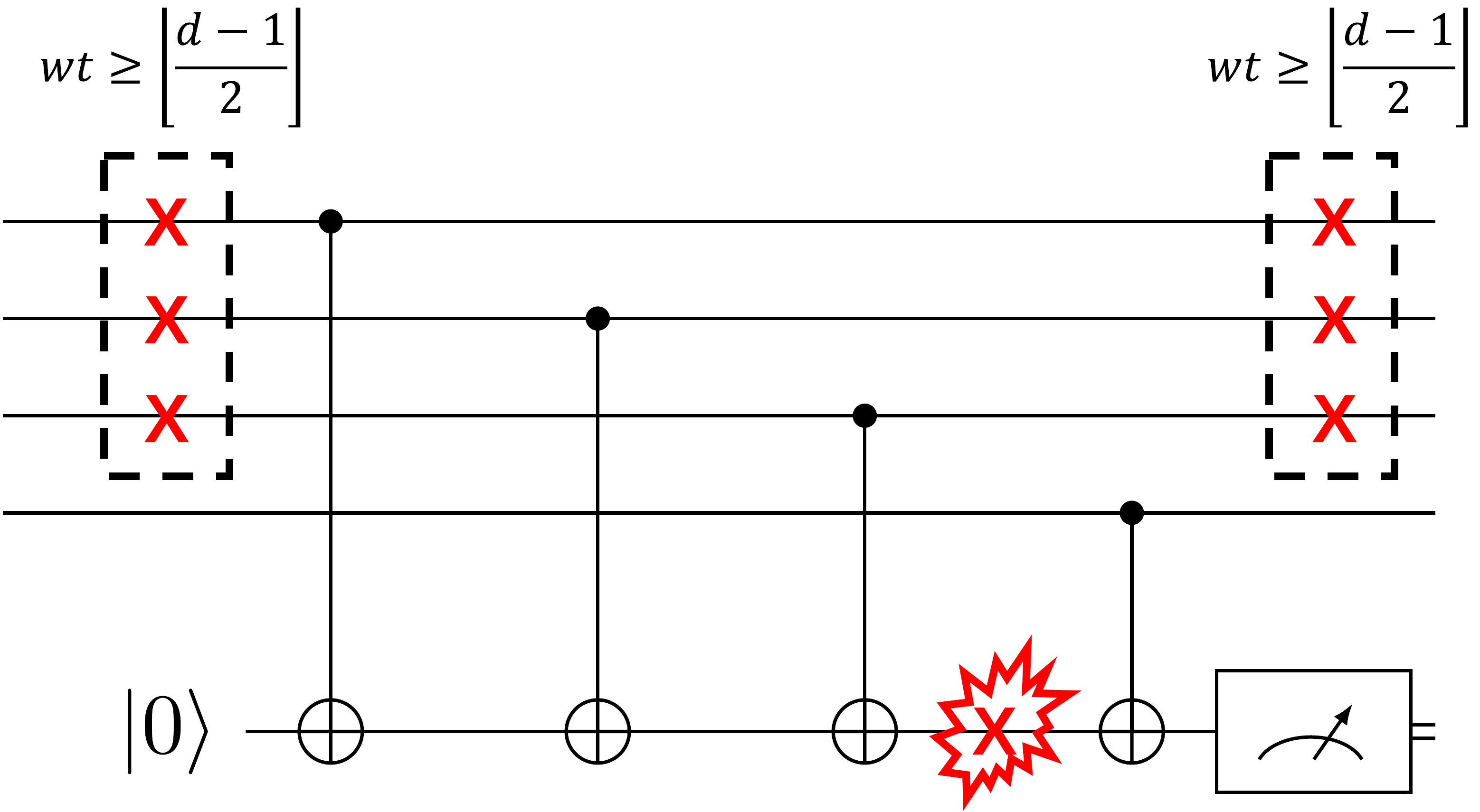}
    \caption{An uncorrectable error would be detected by the verification circuit, but an error on the measurement ancilla flips the $X$ parity, causing the measurement to falsely indicate no error (masked). If this happens after the last measurement that would detect this error, it goes by undetected.}\label{fig:error-case-2}
\end{figure}

The layered construction of the verification circuit simplifies the problem of verification circuit synthesis significantly. 
Instead of synthesizing the verification circuit wholesale, every layer can be constructed individually. Every $\mathcal{E}_i$ gives rise to a new instance of the verification circuit synthesis problem. 
For every $\mathcal{E}_i$ we have to find a set of stabilizer measurements such that every error in $\mathcal{E}_i$ anticommutes with at least one of these measurements.

This layered approach significantly simplifies reasoning about verification circuit constructions by breaking the synthesis down into the synthesis of smaller sub-problems.
This, however, potentially comes at the price of the ``global'' optimality of the full circuit.
Since the different fault sets of the preparation and verification circuits are not necessarily independent, synthesizing the verification circuits separately might lead to certain errors redundantly being verified multiple times.
Hence, in this work, we only consider optimal constructions for each subcircuit.

Next, we will show optimal and heuristic algorithms to tackle these individual synthesis problems.

\subsection{Optimal Approach}\label{sec:optim-verif-circ}
We can frame the optimal verification circuit synthesis problem again as a SAT problem with the following inputs:

\begin{itemize}
\item $\mathcal{C}$: a state-preparation circuit over $n_q$ qubits.
\item $H_X \in \mathbb{F}_2^{m \times n}$: matrix representation of $m$ $X$ stabilizer generators of $\calC$.
\item $H_Z \in \mathbb{F}_2^{(n-m)\times n}$: matrix representation of $n-m$ $Z$ stabilizer generators of $\mathcal{C}$.
\item $n_\mathrm{anc}$: maximum number of ancillas (measurements) used in the verification circuit.
\item $n_\mathrm{CNOT}$: Maximum number of CNOTs used in the verification circuit.
\end{itemize}

Let $\mathcal{E}(\calC)$ denote the set of uncorrectable errors arising from propagating up to $\floor{\frac{d}{2}}$ independent $X$errors and let $E_i$ denote the $\mathbb{F}_2^n$ representation of the $i$-th error in $\mathcal{E}(\calC)$.
The $Z$-generators $\mathcal{Z} = \{Z_1, \cdots, {Z}_{n-m}\}$ form a basis of the $Z$-stabilizers, so we can express any verification circuit as a set of $\mathbb{F}_2^{n-m}$ vectors indicating the contributing generators. 
We introduce the Boolean variables

\[
  \{\mathrm{v}_s^i \mid i \in [n_\mathrm{anc}], s \in [n-m]\}.
  \]

  An assignment to these variables denotes the $i$-th measurement as a linear combination of the $n-m$ $Z$-stabilizer generators.

  The $j$-th entry of the $i$-th measurement can be symbolically computed via

  \[
    \mathrm{meas}^i_j = \left( \bigoplus_{s=1}^{n-m} v_s^i\land H_Z[s, j]\right).
  \]

  Assuming that $H_Z$ is sparse, it is inefficient to convert the $Z$ stabilizers to Boolean constants and then perform these nested XOR operations. 
  Since $a \oplus 0 = a, a \land 1 = a$ and $a\land 0 = 0$, we can simplify the computation by evaluating constants in Boolean functions before encoding the SMT formula. 
  Such simplifications have been shown to lead to dramatic speedups in SMT solving~\cite{shuttyDecodingMergedColorSurface2022}.

  Next, we need to ensure that every error is detected. 
  Recall that an error is detected if it overlaps with a measurement on an odd number of qubits.
  Instead of computing the overlap on all qubits, we can again simplify this constraint by only considering qubits where an error has support. 
  Every error $E_k$ gives rise to a function

  \[
    \mathrm{detected}^{(k, i)} = \left(\bigoplus_{j\in \mathrm{supp}(E_k)}\mathrm{meas}^i_j\right).
  \]

  Every error needs to be detected by at least one measurement, so we impose

  \[
    \bigwedge_{k=1}^{|\mathcal{E}(C)|} \Bigl(\bigvee_{i=1}^{n_\mathrm{anc}}\mathrm{detected}^{(k, i)}\Bigr).
  \]

  We can furthermore restrict the overall number of CNOTs used in the verification circuit with the pseudo-boolean constraint:
  \[
    \Bigl(  \sum_{i=1}^{n_{\mathrm{anc}}}\llbracket \sum_{j=1}^{n_q} \llbracket \mathrm{meas}^i_j \rrbracket  \rrbracket \Bigr) \leq n_\mathrm{CNOT}.
  \]

  Since this encoding takes the maximal number of ancillae and the maximal number of CNOTs as parameters, we can iterate on both parameters to find either ancilla- or gate-optimal verification circuits. 

\subsection{Heuristic Approach}\label{sec:heur-verif-circ}

A part of the difficulty of computing verification circuits is the size of the stabilizer group from which measurements need to be chosen.
Instead of searching through the entire group, we can try to find a suitably large set of \emph{candidate} measurements and try to find a good subset as a verification circuit.
Indeed, if more than one measurement is required to verify a given fault set, the main difficulty is identifying a small set of verification measurements.
The requirement of finding a small set of measurements such that every error is detected by at least one measurement is similar to the \textsc{Set Cover}~\cite{karpReducibilityCombinatorialProblems1972} problem, which is defined as follows:

\begin{problem}
  \problemtitle{\textsc{Set Cover}}
  \probleminput{A set $U$, a collection of subsets $P \subseteq \calP(U)$ and an integer $u$}
  \problemquestion{Is there a set $\mathrm{Cov} \subseteq P$ with $|\mathrm{Cov}| \leq u$ and $\bigcup_{A \in \mathrm{Cov}} A = U$}
\end{problem}

The difference in our setting is that the number of potential stabilizers (covering sets) might be very large. 
We therefore consider this modified version of \textsc{Set Cover}:

\begin{problem}
  \problemtitle{\textsc{Set Cover with Symmetric Difference}}
  \probleminput{A set $U$, a collection of subsets $P \subseteq \calP(U)$ and an integers $u$ and $v$}
  \problemquestion{Define $P_1 = P$ and $P_k = P_{k-1} \cup \{A \bigtriangleup B \mid A, B \in P_{k-1}\}$, where $\bigtriangleup$ denotes the \emph{symmetric difference} of sets. Find $\mathrm{Cov} \subseteq  P_v$ such that $\bigcup_{A \in \mathrm{Cov}}A = U$ and $|\mathrm{Cov}| \leq u$.}
\end{problem}

By identifying $S \in \mathcal{Z}$ with the set $E(S, \calE) = \{F \in \mathcal{E} \mid S\cdot F = -F \cdot S\}$, for two stabilizers $S_1, S_2 \in \mathcal{Z}$ we have $E(S_1\cdot S_2, \calE) = E(S_1, \calE)\bigtriangleup E(S_2, \calE)$, i.e., any error in $\calE$ that anticommutes with both $S_1$ and $S_2$ will not anticommute with their product and therefore will not be detected by the respective measurement.
With this translation, we can phrase an instance of verification circuit synthesis as an instance of the above-defined modified \textsc{Set Cover} problem.

The straight-forward greedy algorithm for \textsc{Set Cover} approximates the optimum quite well~\cite{chvatalGreedyHeuristicSetCovering1979}, so we can use a similar strategy for synthesizing a verification circuit.
The idea is to start with the stabilizer generators and iteratively expand the set of candidate measurements by building products of stabilizers in the candidate set.
In each iteration, we can try to find a small set of stabilizers such that all problematic errors are detected.
If increasing the candidate set does not lead to a verification circuit with fewer stabilizer measurements, we stop.
Alternatively, if the size of the candidate set reaches some pre-defined limit, we can take the best solution found until this point and terminate.

This greedy algorithm is formalized in~\Cref{alg:heur-ver}. This algorithm assumes a function \verb|set_cover(U, P)| that computes a solution to the \textsc{Set Cover} problem for the sets $U$ and $P \subseteq \calP$ which can be done in polynomial time using the greedy algorithm.

\begin{algorithm}
\caption{Heuristic Verification Circuit}\label{alg:heur-ver}
\KwData{Set of errors $\mathcal{E}$, $Z$-stabilizer generators $\mathcal{Z}$, maximal candidate set size $c_\mathrm{max}$}
\KwResult{Set of stabilizer measurements $\mathcal{Z}_\text{ver}$ that detect $\mathcal{E}$}
$P \gets \{E(S, \mathcal{E}) \mid S \in \mathcal{Z}\}$\;
$\mathrm{Cov} \gets $ \texttt{set\_cover}$(P, \calE)$\;
\While{$|P| \leq c_\mathrm{max}$} {
  $\mathrm{Cov}_\text{new} \gets $ \texttt{set\_cover}$(P, E(S))$\;
  \If{$|\mathrm{Cov}_\text{new}| = \mathrm{Cov}$} {
    break\;
  }
  $P \gets P \cup \{A \bigtriangleup B \mid A, B \in P\}$\;
  $\mathrm{Cov} \gets \mathrm{Cov}_\text{new}$\;
}
$\mathcal{Z}_\text{ver} \gets$ Stabilizers corresponding to $\mathrm{Cov}$\;
\Return $\mathcal{Z}_\text{ver}$\;
\end{algorithm}

Let us make two remarks about the implementation of~\Cref{alg:heur-ver}:

\begin{itemize}
\item When solving \textsc{Set Cover} in lines $2$ and $4$ there might be multiple choices for every greedy choice of set. 
In this case, we preferably pick the set that corresponds to lower-weight stabilizers.
\item Extracting the stabilizers from the set cover in line $9$ can be done efficiently by maintaining a map from sets to stabilizer products. 
\end{itemize}

In general, the candidate set can become large very quickly as, in the worst case, it doubles in size during every iteration of the while-loop. 
For this reason,~\Cref{alg:heur-ver} takes the parameter $c_\mathrm{max}$ to limit the growth of this set.

By design,~\Cref{alg:heur-ver} only aims at minimizing the number of measurements in the verification circuit and not the weight of those measurements.
This approach is appropriate when minimizing the overall number of measurements but might result in a few high-weight stabilizer measurements. 
This problem can be somewhat mitigated by completely removing high-weight measurements from the set of measurement candidates.
However, a few high-weight measurements could actually be less costly in terms of the number of CNOTs than many low-weight measurements. 
The preferred choice might, for example, depend on the noise characteristics of the respective hardware under consideration.

After an appropriate set of stabilizers has been found using~\Cref{alg:heur-ver}, we can try to minimize the weight of these stabilizers. 
If there are stabilizers $\mathcal{Z}_\mathrm{red} = \{S \in \langle \mathcal{Z} \rangle \mid \mathcal{E}(S, \mathrm{E}) = \emptyset\}$, i.e., stabilizers that commute with all errors, then they form a group $\langle \mathcal{Z}_\mathrm{red} \rangle$. 
Given a verification stabilizer $S_\mathrm{ver}$ and any $S_\mathrm{red} \in \langle \mathcal{Z}_\mathrm{red} \rangle$ we can replace $S_\mathrm{ver}$ by $S_\mathrm{ver}\cdot S_\mathrm{red}$ in the verification circuit without changing which errors are detected. %
Another way to phrase this is to say that replacing $S_\mathrm{ver}$ by any element in the coset $S_\mathrm{ver}\cdot \langle  \mathcal{Z}_\mathrm{red} \rangle$ yields a valid verification circuit.
Give the measurements output by~\Cref{alg:heur-ver} we can therefore try to find the minimal representatives in their respective cosets, i.e., the so-called coset leader.
As noted above, computing the minimal-weight representatives of these cosets is NP-hard~\cite{berlekampInherentIntractabilityCertain1978}. 
Apart from constructing the standard array (i.e., listing all coset elements), there are methods using techniques from computer algebra~\cite{borges-quintanaComputingCosetLeaders2010,borges-quintanaComputingCosetLeaders2014}, which still require exponential memory in the worst case~\cite{mayrComplexityResultsPolynomial1997}. 
As a post-processing step, we can apply a heuristic search for the coset leaders of our verification stabilizers. 
One can also relegate this step to an SMT solver. 
For small codes, this can be expected to be efficient even if the global optimal verification circuit cannot be computed with the optimal method in an adequate amount of time.

\section{Fully Fault-Tolerant Non-Deterministic State Preparation}
\label{sec:fully-fault-tolerant}

Until now, we have only considered verification circuits for one type of error, i.e., $X$ errors in preparing the $\ket{0}_L^{\otimes k}$. But even though $Z$ errors cannot cause a logical error on $\ket{0}_L^{\otimes k}$, propagating $Z$ errors through the state preparation circuit still violates the fault-tolerant requirement of \Cref{def:ft}. A straightforward fix to this is generating verification circuits $\mathrm{Ver}_X$ and $\mathrm{Ver}_Z$ according to \Cref{sec:ver-circ} that check for $X$ and $Z$ errors respectively and combine them to the verification circuit $\mathrm{Ver} = \mathrm{Ver}_X \circ \mathrm{Ver}_Z$. This does not work however, because $X$ ($Z$) errors on one of the ancillae in $\mathrm{Ver}_X$ ($\mathrm{Ver}_Z$) can propagate to a higher weight error on the data qubits. This is commonly referred to as a \emph{hook error}~\cite{dennisTopologicalQuantumMemory2002}.

Unfortunately, there is no way to fix this predicament by introducing further stabilizer measurements since the new measurements can themselves be faulty. Instead, the hook errors have to be checked for specifically. Flag fault-tolerant stabilizer measurements have been introduced to solve this exact problem~\cite{chamberlandFlagFaulttolerantError2018,chaoFlagFaultTolerantError2020}. The idea is to add further \enquote{flag} ancillae that interact with the stabilizer measurement ancilla via a pair of CNOTs that---in the absence of errors---act as identity. However, if any error occurs on the measurement ancilla, it propagates to at least one of the flags and is detected. See \Cref{fig:flag-ft} for an example of a weight 4 flag fault-tolerant measurement.

\begin{figure}[t]
  \centering
  \begin{tikzpicture}
    \begin{yquant}
      qubit {} q[4];

      hspace {0.5cm} q;      
      [after=q]
      qubit {$\ket{0}$} a[1];
      cnot a[0] | q[0];
      [after=a]
      qubit {$\ket{+}$} flag[1];
      cnot a[0] | flag[0];
      cnot a[0] | q[1];
      cnot a[0] | q[2];
      cnot a[0] | flag[0];
      cnot a[0] | q[3];
      measure  a;
      measure {$X$} flag;
    \end{yquant}
  \end{tikzpicture}
  \caption{Flag fault-tolerant measurement of a weight-$4$ stabilizer~\cite{chamberlandFlagFaulttolerantError2018}}
  \label{fig:flag-ft}
\end{figure}

Flag fault tolerance was conceived as a way to give decoders additional information to account for hook errors, which can lead to complicated decision trees for decoders depending on the exact pattern of flags that measure an error. We require no such finesse and opt to include the flag measurements as further measurements upon which we post-select. If any flag is measured in $-1$, this indicates that an error was propagated from an ancilla to the data, and we discard the state. 

Not all stabilizers necessarily have to be measured flag fault-tolerantly. If we measure $\mathrm{Ver}_X$ first, we can account for any hook errors introduced by this circuit by adapting the measurements for $\mathrm{Ver}_Z$. Only $\mathrm{Ver}_Z$, then, has to be measured using flags. Depending on the code and state preparation circuit, reversing the order of verification circuits might be better in terms of CNOTs. Still, only the second verification circuit has to be flag fault-tolerant.
With this flag-fault-tolerant measurement scheme for the second error type, we arrive at the proposed scheme sketched in \Cref{fig:full-ft}.

Depending on the weight of the measurements in the first verification circuit, the number of additional hook errors might be pretty significant.
These additional errors might necessitate complex measurements in the subsequent verification circuit.
In that case, it might be better to also flag the first layer.

At this point, we have to revisit the issue of optimal circuits for the verification circuit synthesis. %
The introduction of flag-fault-tolerant measurements complicates the gate-optimality of the verification measurements even for the individual fault sets.
First, no generally optimal flag scheme is known~\cite{chamberlandFlagFaulttolerantError2018,chaoFlagFaultTolerantError2020}.
Secondly, there is also a trade-off between the size weight of the measured stabilizers and the size of the corresponding flag circuit.
It might, therefore, be better to measure more low-weight stabilizers with low flag overhead instead of high-weight stabilizers with high flag overhead. Note that depending on the minimum-weight stabilizer generators of the given code, the use of flag qubits can not always be avoided.

We leave the problem of solving this global optimization problem for future work. In this work, we focus on providing locally optimal state preparation and verification circuits without considering flag overhead.

\section{Numerical Evaluations} %
\label{sec:eval}

The implementation of the methods presented in this work is publicly available as part of the Munich Quantum Toolkit~\cite{willeMQTHandbookSummary2024, berentMQTQECC2025}. 
To evaluate the efficacy of the proposed methods, we have synthesized state preparation circuits for various CSS codes.

\subsection{Circuit Construction}

\begin{table*}[t]
  \centering
    \caption{Circuit metrics for fault-tolerant state preparation circuits obtained using our optimal (SAT) and heuristic approach when optimizing for CNOT count, and comparison with the recently proposed RL method~\cite{zenQuantumCircuitDiscovery2024}.}\label{tab:circuits_cnot}
    \resizebox{\textwidth}{!}{
      \begin{threeparttable}      
  \begin{tabular}{l c | c c c |c c c| c c c | c c c}
    \multirow{3}{*}{Code}                                                                                                    & \multirow{2}{*}{State} & \multicolumn{6}{c|}{State Preparation Circuit} & \multicolumn{6}{c}{Verification Circuit}                                                                                      \\ 
                                                                                                                             &                        & \multicolumn{3}{c}{\#CNOTs}                    & \multicolumn{3}{c|}{Depth} & \multicolumn{3}{c}{\#CNOTs} & \multicolumn{3}{c}{\#Measurements}                                 \\
                                                                                                                             &                        & SAT                                            & Heuristic                  & RL                          & SAT & Heuristic & RL & SAT & Heuristic & RL & SAT & Heuristic & RL \\ \midrule
    \multirow{2}{*}{$\llbracket 15, 1, 3 \rrbracket$ Reed-Muller (tetrahedral) code~\cite{steaneQuantumReedMullerCodes1999}} & $\ket{0}_L$            & 22                                             & 22                         & 22                          & 4   & 4         & 7  & 3   & 3         & 3  & 1   & 1         & 1  \\
                                                                                                                             & $\ket{+}_L$            & \textbf{23}                                    & \textbf{23}                & 24                          & 5   & 5         & 6  & 7   & 7         & 7  & 1   & 1         & 1  \\
    \multirow{2}{*}{$\llbracket 9, 1, 3 \rrbracket$ Shor Code~\cite{shorSchemeReducingDecoherence1995}}                      & $\ket{0}_L$            & 8                                              & 8                          & 8                           & 5   & 6         & 5  & 3   & 3         & 3  & 1   & 1         & 1  \\
                                                                                                                             & $\ket{+}_L$            & 6                                              & 6                          & 6                           & 2   & 2         & 2  & 0   & 0         & 0  & 0   & 0         & 0  \\
    $\llbracket 9, 1, 3 \rrbracket$ Rotated Surface Code~\cite{bombinOptimalResourcesTopological2007}                        & $\ket{0}_L$            & 8                                              & 8                          & 8                           & 5   & 6         & 3  & 3   & 4         & 3  & 1   & 1         & 1  \\
    $\llbracket 15, 7, 3 \rrbracket$ Hamming Code~\cite{steaneSimpleQuantumErrorcorrecting1996}                              & $\ket{0}_L^7$          & -                                              & 22                         & -                           & -   & 12        & -  & 6*  & 7         & -  & 2   & 2         & -  \\
    $\llbracket 17, 1, 5 \rrbracket$ 2D Color Code~\cite{bombinTopologicalQuantumDistillation2006}                           & $\ket{0}_L$            & -                                              & \textbf{23}                & 25                          & -   & 11        & 7  & 29* & 32        & -  & 6*  & 5         & -  \\
    $\llbracket 19, 1, 5 \rrbracket$ 2D Color Code~\cite{bombinTopologicalQuantumDistillation2006}                           & $\ket{0}_L$            & -                                              & 27                         & -                           & -   & 8         & -  & 33* & 38        & -  & 6*  & 6         & -  \\
    $\llbracket 25, 1, 5 \rrbracket$ Rotated Surface Code~\cite{bombinOptimalResourcesTopological2007}                       & $\ket{0}_L$            & -                                              & 28                         & -                           & -   & 13        & -  & 34* & 34        & -  & 5*  & 5         & -  \\    
  \end{tabular}
  \begin{tablenotes}
    \tiny
  \item Entries marked with \enquote{*} denote circuits where the optimal verification circuit was generated for the heuristic state preparation circuit.
  \item Entries with a \enquote{-} denote instances where the respective method failed to provide a circuit within the timeframe of 24h.
  \item Entries in bold denote instances where the obtained circuits using our methods yields improved metrics.
\end{tablenotes}
\end{threeparttable}
}
\end{table*}

\begin{table*}[t]
  \centering
    \caption{Circuit metrics for fault-tolerant state preparation circuits obtained using our optimal (SAT) and heuristic approach when optimizing for circuit depth, and comparison with the recently proposed RL method~\cite{zenQuantumCircuitDiscovery2024}.}\label{tab:circuits_depth}
    \resizebox{\textwidth}{!}{
      \begin{threeparttable}
  \begin{tabular}{l c | c c c |c c c| c c c | c c c}
    \multirow{3}{*}{Code}                                         & \multirow{2}{*}{State} & \multicolumn{6}{c|}{State Preparation Circuit} & \multicolumn{6}{c}{Verification Circuit}                                                                                             \\ 
                                                                                     &               & \multicolumn{3}{c}{\#CNOTs} & \multicolumn{3}{c|}{Depth} & \multicolumn{3}{c}{\#CNOTs} & \multicolumn{3}{c}{\#Measurements}                                         \\
                                                                                     &               & SAT                         & Heuristic                  & RL                          & SAT        & Heuristic  & RL & SAT & Heuristic & RL & SAT & Heuristic & RL \\ \midrule
        \multirow{2}{*}{$\llbracket 15, 1, 3 \rrbracket$ Reed-Muller (tetrahedral) code~\cite{steaneQuantumReedMullerCodes1999}} & $\ket{0}_L$   & 22                          & 22                         & 22                          & \textbf{4} & \textbf{4} & 7  & 3   & 7         & 3  & 1   & 2         & 1  \\
                                                                                                                                 & $\ket{+}_L$   & \textbf{23}                 & \textbf{23}                & 24                          & \textbf{5} & \textbf{5} & 6  & 7   & 7         & 7  & 1   & 1         & 1  \\
        \multirow{2}{*}{$\llbracket 9, 1, 3 \rrbracket$ Shor Code~\cite{shorSchemeReducingDecoherence1995}}                      & $\ket{0}_L$   & 8                           & 9                          & 8                           & \textbf{3} & \textbf{3} & 5  & 3   & 5         & 3  & 1   & 2         & 1  \\
                                                                                                                                 & $\ket{+}_L$   & 6                           & 6                          & 6                           & 2          & 2          & 2  & 0   & 0         & 0  & 0   & 0         & 0  \\
        $\llbracket 9, 1, 3 \rrbracket$ Rotated Surface Code~\cite{bombinOptimalResourcesTopological2007}                        & $\ket{0}_L$   & 8                           & 8                          & 8                           & 3          & 3          & 3  & 3   & 4         & 3  & 1   & 1         & 1  \\
        $\llbracket 15, 7, 3 \rrbracket$ Hamming Code~\cite{steaneSimpleQuantumErrorcorrecting1996}                              & $\ket{0}_L^7$ & 23                          & 22                         & -                           & 4          & 4          & -  & 8*  & 10        & -  & 2*  & 2         & -  \\
        $\llbracket 17, 1, 5 \rrbracket$ 2D Color Code~\cite{bombinTopologicalQuantumDistillation2006}                           & $\ket{0}_L$   & -                           & \textbf{23}                & 25                          & -          & \textbf{5} & 7  & 28* & 31        & -  & 6*  & 5         & -  \\
        $\llbracket 19, 1, 5 \rrbracket$ 2D Color Code~\cite{bombinTopologicalQuantumDistillation2006}                           & $\ket{0}_L$   & -                           & 27                         & -                           & -          & 5          & -  & 33* & 43        & -  & 5*  & 6         & -  \\
        $\llbracket 25, 1, 5 \rrbracket$ Rotated Surface Code~\cite{bombinOptimalResourcesTopological2007}                       & $\ket{0}_L$   & -                           & 28                         & -                           & -          & 5          & -  & 35* & 38        & -  & 4*  & 5         & -  \\
  \end{tabular}
  \begin{tablenotes}
    \tiny
  \item Entries marked with \enquote{*} denote circuits where the optimal verification circuit was generated for the heuristic state preparation circuit.
  \item Entries with a \enquote{-} denote instances where the respective method failed to provide a circuit within the timeframe of 24h.
  \item Entries in bold denote instances where the obtained circuits using our methods yields improved metrics.
\end{tablenotes}
\end{threeparttable}
}
\end{table*}

\begin{table*}[t]
  \centering
  \caption{Circuit metrics for the best circuits found using the proposed methods for the generalized non-deterministic fault-tolerant state preparation scheme.}\label{tab:circuits}
  \resizebox{\linewidth}{!}{
  \begin{tabular}{l c | c c c c | c c}
    \multirow{2}{*}{Code}                                                                                                                     &   \multirow{2}{*}{State}       & \multicolumn{4}{c|}{This work} & \multicolumn{2}{c}{Projective Initialization} \\           
                                                                                                                             &               & \#Ancillas w/o reuse & \#Ancillas with reuse & \#CNOTs & Depth & \#CNOTs worst case & \#CNOTs best case \\                       \midrule
    \multirow{2}{*}{$\llbracket 15, 1, 3 \rrbracket$ Reed-Muller (tetrahedral) code (see~\Cref{fig:tetrahedral_circuits})~\cite{steaneQuantumReedMullerCodes1999}} & $\ket{0}_L$   & 1                    & 1                     & 25      & 8     & 96                 & 64                \\                                   
                                                                                                                             & $\ket{+}_L$   & 5                    & 2                     & 42      & 14     & 120                & 80                \\
    $\llbracket 9, 1, 3 \rrbracket$ Rotated Surface Code~\cite{bombinOptimalResourcesTopological2007}                        & $\ket{0}_L$   & 1                    & 1                     & 11      & 6     & 36                 & 24                \\                                              
    $\llbracket 12, 2, 4 \rrbracket$ \enquote{Carbon} code (see~\Cref{fig:carbon})~\cite{m.p.dasilvaDemonstrationLogicalQubits2024}                  & $\ket{0}_L^2$ & 3                    & 2                     & 30      & 16    & 120                & 90                \\                                              
    $\llbracket 15, 7, 3 \rrbracket$ Hamming Code (see~\Cref{fig:hamming})~\cite{steaneSimpleQuantumErrorcorrecting1996}                              & $\ket{0}_L^7$ & 2                    & 1                     & 28      & 11    & 96                 & 64                \\                                             
    $\llbracket 17, 1, 5 \rrbracket$ 2D Color Code (see~\Cref{fig:ket0-17-1-5})~\cite{bombinTopologicalQuantumDistillation2006}                           & $\ket{0}_L$   & 13                   & 4                     & 71      & 22    & 180                & 108               \\                                             
    $\llbracket 19, 1, 5 \rrbracket$ 2D Color Code (see~\Cref{fig:ket0-19-1-5})~\cite{bombinTopologicalQuantumDistillation2006}                           & $\ket{0}_L$   & 14                   & 4                     & 120     & 25    & 210                & 126               \\
    $\llbracket 25, 1, 5 \rrbracket$ Rotated Surface Code~\cite{bombinOptimalResourcesTopological2007}                       & $\ket{0}_L$   & 25                   & 4                     & 127     & 31    & 200                & 120               \\
  \end{tabular}
  }
\end{table*}

In this section, we focus on the investigation of circuit metrics.

\subsubsection{Circuits for Distance 3 and 5 Codes}
\label{sec:circuits-distance-3}

We generate non-deterministic fault-tolerant state preparation circuits for various $d=3$ and $d=5$ CSS codes using both the SMT and heuristic methods proposed in~\Cref{sec:state-prep} and~\Cref{sec:ver-circ}. 
This allows us to investigate the scalability of the exact approach as well as how close the heuristic approach is to the optimum. 
In addition to that, we also compare the results to solutions generated in Ref.~\cite{zenQuantumCircuitDiscovery2024} (as reviewed in~\Cref{sec:related-work}, a synthesizer based on reinforcement learning).
Note that on our machine the RL agent does not manage to find state preparation circuits within 24 hours for codes with 12 or more physical qubits. 
In this case, we used the best circuit from the examples provided at \hyperlink{https://github.com/remmyzen/rlftqc}{https://github.com/remmyzen/rlftqc}. 

\Cref{tab:circuits_cnot} and~\Cref{tab:circuits_depth} provide a summary of the circuits obtained for both gate- and depth-optimization~\footnote{
Note that the RL method does not minimize  circuit depth and hence may produce different circuits of varying depth for multiple attempts.
Here, we use the instances that have the lowest depth for a minimal CNOT count}, respectively.
The tables list the CNOT count and depth of the non-fault-tolerant state preparation and verification circuits obtained from the proposed optimal and heuristic methods as well as the RL agent. 
More precisely, the first two columns denote the code instance and the prepared logical state, while the remaining columns denote the metrics of the obtained state preparation and verification circuits. 
``SAT'' denotes the exact results, ``heuristic'' the heuristic results, and RL the solution obtained by 
 the RL agent of Ref.~\cite{zenQuantumCircuitDiscovery2024}. 

For this comparison, we only focussed on verification circuits for the primary type of error ($X$ errors for the preparation of $\ket{0}_L^{\otimes k}$ and $Z$ errors for the preparation of $\ket{+}_L^{\otimes k}$).
Note that the number of measurements is equal to the number of additional ancilla qubits needed if qubits are not reused.
Because hook errors are not considered for this comparison---and therefore, no additional flag qubits are needed---a single ancilla could be reused to perform all measurements.

These results show that for small codes, all methods produce similar results. 
For the cases where all methods produce the same circuit, however, the optimal method also yields an automatic proof that the obtained circuit is, in fact, gate- or depth-optimal. 
For larger codes ($n\geq 17$), both the optimal synthesis method and the RL method fail to yield a result within a 24-hour timeout. 
For the optimal synthesis methods specifically, the state preparation times out.
The verification circuit synthesis scales much more favorably. 
Moreover, the proposed approach clearly outperforms the RL synthesizer.

Interestingly, for the check matrices we investigated, optimizing for depth never comes at the cost of additional CNOT gates.
This does not mean, however, that this cannot happen in general as one can construct check matrices where our implementation needs more CNOTs for the depth-optimal circuit.

For codes with more qubits ($n \geq 17$), the heuristic state preparation circuit starts to outperform the RL synthesis regarding CNOT count and circuit depth. 
For larger codes, the SMT and RL synthesis fail (indicated by ``$-$'' in the table) to produce any state preparation circuit within 24 hours. 
The heuristic synthesis, however, generates all circuits in under one second. 
Even in cases where the SMT solver cannot find an optimal state preparation circuit, it is still possible to find the optimal verification circuit (for the state preparation circuit generated by the heuristic synthesis). 
From these tables, we can conclude that synthesizing state preparation circuits using our greedy heuristic approach in combination with the optimal verification circuit synthesis can produce competitive non-deterministic fault-tolerant state preparation circuits even for moderately-sized codes.

\begin{table*}[t]
  \centering
  \caption{Verification Circuits for the $\ket{0}_L$ preparation circuit for the $\llbracket 31,1,7 \rrbracket$ color code in~\Cref{fig:d7-color}.}
  \label{tab:d7_prep}
  \begin{tabular}{l | c c c |  c c c | c c c}
    \multirow{2}{*}{Fault Set} & \multicolumn{3}{c|}{Optimal} & \multicolumn{3}{c|}{Heuristic} & \multicolumn{3}{c}{Naive} \\
     & Weights &  \#Flags & \#CNOTs & Weights & \#Flags & \#CNOTs  &  Weights & \#Flags & \#CNOTs  \\ \midrule
    $\calE_1^X$ & $[7,4,7]$ & 9 & 36 & $[12, 8]$ & 10 & 40 & $4\times 11, 8 \times 3$ & 20 & 112 \\
    $\calE_2^X$ & $[8,7,7,7]$ & 16 & 61 & $[11,8,7,7]$ & 18 & 69 & $4\times 11, 8 \times 3$ & 20 & 112 \\
    $\calE_3^X$ & - & - & - & $[4,8,4,4,7,8,7]$ & 19 & 80 & $4\times 11, 8 \times 3$ & 20 & 112 \\ \midrule
    $\calE_1^Z$ & $[12, 8]$ & 10 & 40 & $[12, 12]$ & 12 & 48 & $4\times 11, 8 \times 3$ & 20 & 112 \\
    $\calE_2^Z$ & $[8,8,8,8]$ & 16 & 64 & $[12,8,4,8,24]$ & 26 & 108 & $4\times 11, 8 \times 3$ & 20 & 112 \\
    $\calE_3^Z$ & - & - & - & $[8,4,8,8,12,8]$ & 23 & 94 & $4\times 11, 8 \times 3$ & 20 & 112 \\
  \end{tabular}
\end{table*}

Next, we use our methods to generate fully fault-tolerant circuits as described in~\Cref{sec:fully-fault-tolerant}.
Again, we give the optimal synthesis method a 24-hour timeout. 
We use the heuristic method if the solver does not find a solution within this timeframe.
All verification circuits were constructed gate-optimally for the $d=3$ and $d=5$ circuits.
Some of the non-fault-tolerant state preparation circuits were generated by the heuristic, either because the optimal synthesis method times out or because the heuristic circuit yielded a verification circuit that leads to an overall improved CNOT count.
\Cref{tab:circuits} shows circuit metrics for gate-optimized circuits found using this approach.
This table shows explicitly how many additional ancillas are required to perform the verification circuit, the number of CNOTs needed for the entire circuit, and the circuit depth.

\Cref{tab:circuits} furthermore lists the number of CNOTs required for the general scheme where the state is prepared through initialization of the data qubits in the  all $\ket{0}$ ($\ket{+}$) state and subsequent projective measurements of the stabilizers.
Note that this state preparation procedure only requires measuring one type of stabilizer since, for CSS codes, the initial product state is already stabilized by all $X$ or $Z$ stabilizers of the code.
Since measuring the stabilizers non-deterministically projects the state into the $+1$ or $-1$ eigenspace, multiple rounds of measurements are required to determine the correct sign of the stabilizers.
For the codes considered here this requires $d$ measurement rounds and in the best case $\floor{\frac{d}{2}}+1$.
\Cref{tab:circuits} lists the CNOT count for both cases.
The numbers for the projective initialization do not consider hook errors and flag qubits, which would add further CNOT gates to these circuits.
Except for the distance five surface code, the proposed procedure requires fewer CNOTs than even the best-case scenario for the projective initialization (and even fewer for all cases if hook errors and flag qubits are taken into account) and significantly fewer CNOTs than the worst-case.
The benefit of projective initialization is, however, that the state is obtained deterministically, which may be favorable for larger code instances in general.

The circuits used for~\Cref{tab:circuits} are partially listed in~\Cref{sec:appendix-circuits}.
All circuits listed in~\Cref{tab:circuits} are also publicly available~\cite{willeMQTHandbookSummary2024, pehamDataAutomatedSynthesis2025}.
An example of a fault-tolerant state preparation of $\ket{0}_L$ for a distance $5$ code is depicted in~\Cref{fig:ket0-17-1-5}.
There, the verification circuit accounts for a large part of the entire circuit compared to the non-fault-tolerant state preparation subcircuit.

\subsubsection{Scaling to Higher-Distance Codes}
\label{sec:scal-high-dist}

Although the optimal verification circuit synthesis approach was sufficient to synthesize circuits for the $d=3$ and $d=5$ codes considered in the previous section, naturally, there are limits to how far this method can be pushed.
When scaling to higher-distance codes, the stabilizer groups, the number of qubits, and the size of the required state preparation circuits grow.
In turn, the size of the fault sets that need verification grows rapidly.
This directly impacts the SAT encoding since the size of the encoding grows as $O(|S|\cdot |\calE|)$, where $S$ is a generating set of the stabilizer group and $\calE$ is the set of errors to be verified.
Given the layered construction of the verification circuit as illustrated in~\Cref{fig:full-ft}, and that the sizes of the fault sets increase for later layers, the optimal synthesis algorithm can still generate the early layers of the verification circuit.

Consider the state preparation circuit for the zero state of the $\llbracket 31,1,7 \rrbracket$ color code.
A state preparation circuit with $46$ CNOTs and depth $5$ generated by the greedy algorithm is given in~\Cref{fig:d7-color}.
To verify this state preparation circuit, we must construct three layers of $X$ verification and three layers for the $Z$ verification.

For $1$ and $2$ independent errors in the state preparation circuit, we obtain fault sets of sizes
\begin{align*}
  &|\calE_1^X| = 16 \quad &|\calE_2^X| = 443 \\
  &|\calE_1^Z| = 14 \quad &|\calE_2^Z| = 315,
\end{align*}
which are still small enough to attempt optimal synthesis for.

\Cref{tab:d7_prep} compares the verification circuits obtained for each fault set by the optimal synthesis, heuristic synthesis, and a naive verification circuit constructed by measuring all stabilizers of the code.
Given the relatively high number of measurements in the $X$ verification, we have opted to flag the measurements in both layers instead of adding potential hook errors to the fault sets for the $Z$ verification.
The table shows the stabilizer weights of the stabilizers measured in the verification circuit, the number of flag qubits needed to make these measurements flag-fault-tolerant, and the total number of CNOTs required to implement this circuit.
In total, using the optimal verification circuit for the first two layers and the circuit obtained by the heuristic method for the third layer in the respective subcircuits yields a verification circuit with $375$ CNOTs, $119$ ancilla qubits and depth $43$.
In comparison, the naive verification circuit has $720$ CNOTs, $243$ ancilla qubits, and a depth of $41$.
The lower depth comes from the fact that the measured stabilizers in the naive verification circuit generally have a lower weight than those in the synthesized measurements.

From \Cref{tab:d7_prep} we can see that quite a large overhead is incurred by flagging the stabilizer measurements.
The weight-$24$ measurement the heuristic synthesis finds for $\calE_2^X$ is particularly bad in this regard.
This suggests that for higher-distance codes with potential high-weight measurements, including a penalty for the stabilizer weight might yield better results. Even a relatively small increase in stabilizer weight could lead to a rise in the number of required flags and, thus, the overall number of CNOTs.
For the $\llbracket 31, 1, 7 \rrbracket$ color code, this effect is even more pronounced since the weight $4$ measurements require only a single flag qubit and two CNOTs to flag. For example, combining $3$ weight-$4$ measurements into one weight $12$ measurements would require $6$ flags instead of $3$ and, therefore, $12$ CNOTs instead of $6$.
As we focus on smaller codes that could be realized on near-term devices, we leave this avenue of investigation open for future work.

\subsection{Logical Error Rate Simulation}

To show that the synthesized circuits indeed yield the expected logical error rates, we have simulated them using the open-source stabilizer simulator Stim~\cite{gidneyStimFastStabilizer2021} under a standard circuit-level depolarizing noise model as considered in previous works~\cite{chamberlandFlagFaulttolerantError2018,heussenMeasurementFreeFaultTolerantQuantum2024}. 
Noisy gates are modeled as ideal gates followed by a depolarizing noise channel parameterized by an error probability $p$.
Since the considered circuits are composed only of two-qubit gates, we define the two-qubit depolarizing channel:
\begin{align*}
  \epsilon_2(\rho) &= (1-p)\rho + \frac{p}{15} \sum_{E \in \mathcal{E} }E \rho E,
\end{align*}

where $\mathcal{E} = \{P_1\otimes P_2 \mid P_1,P_2 \in \{I,X,Y,Z\} \}\setminus \{I\otimes I\}$. 
In addition, there is a probability of $2p/3$ that a qubit will erroneously be initialized in the $-1$ eigenstate of the respective basis and measurements have a $2p/3$ probability of being flipped. 
If a qubit is idle during one layer of CNOTs, it is also subject to single-qubit depolarizing ``idling'' noise.
What constitutes a layer of CNOTs depends on the computational model. We consider the two extreme cases: either all parallel CNOTs are executed in one layer, or only one CNOT can be performed in each layer.
By default, we set $p_\mathrm{idle}=p/100$, but since idle noise can have a significant impact on the performance of the state preparation circuits, we also consider other settings as described below.

Simulation runs for which any of the verification measurements indicate an error are discarded. 
For the remaining states, we calculate the logical error rate by measuring all data qubits in the $Z$ basis, performing one round of ideal syndrome extraction and correction using a lookup-table decoder. 
A logical error is registered when the corrected classical bitstring anticommutes with any of the $Z$ operators of the code. 
For each physical error rate $p$, we sample until we have collected $500$ logical errors. 
For our simulations, this ensures that the uncertainties for the logical error rates are below $2\%$ for all runs.

\begin{figure*}[t]
   \captionsetup{
    justification=raggedright,
}
  \centering
  \begin{subfigure}[t]{.8\linewidth}
    \centering
    \includegraphics[width=\linewidth]{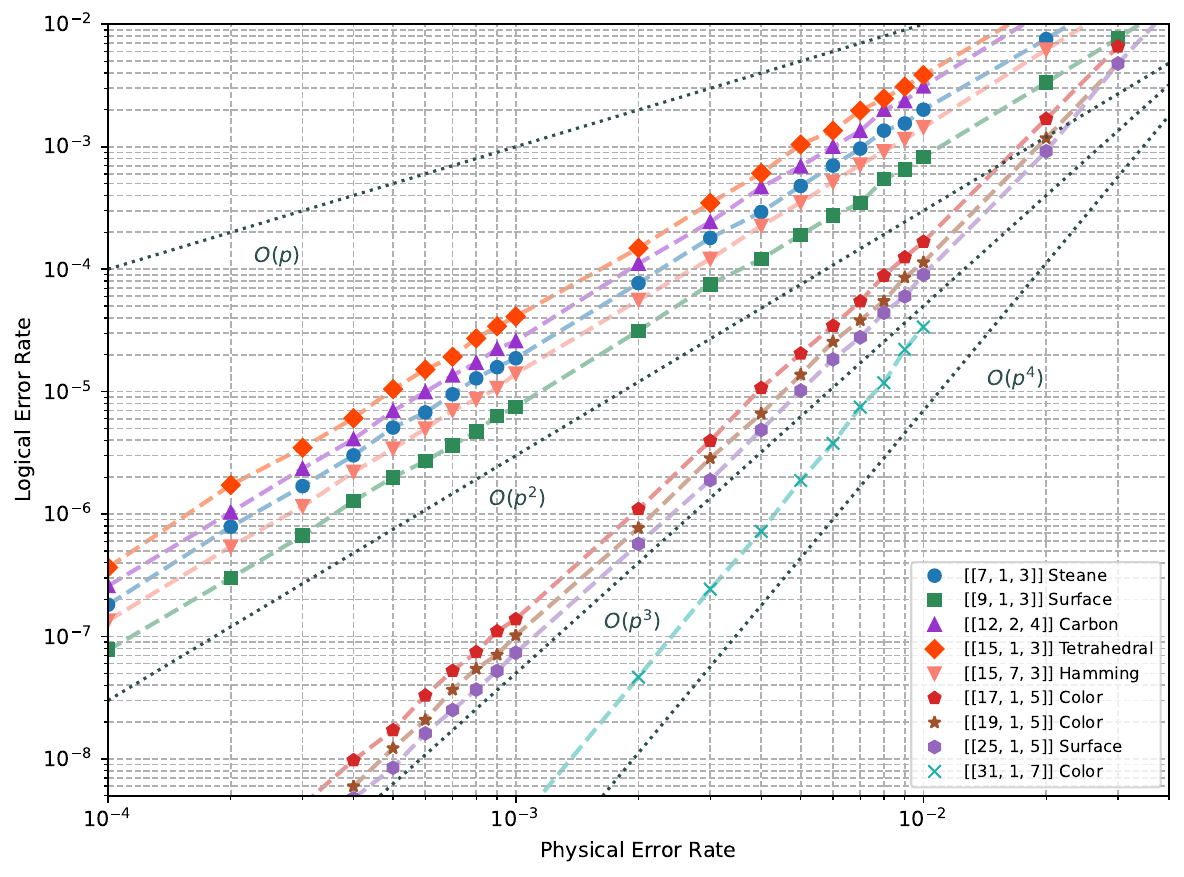}
  \end{subfigure}

  \begin{subfigure}[t]{.8\linewidth}
    \centering
    \includegraphics[width=\linewidth]{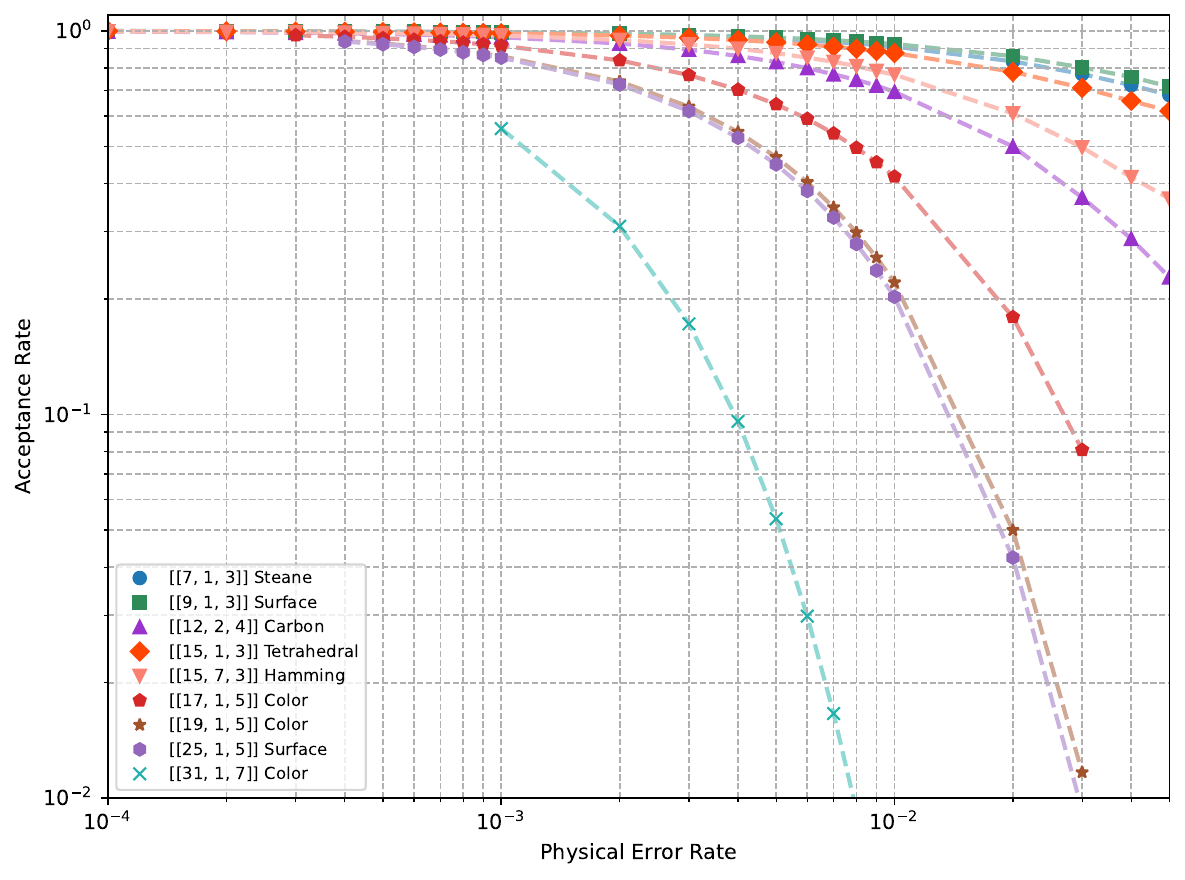}
  \end{subfigure}
  \caption{Logical error rate and acceptance rate for non-deterministic fault-tolerant state preparation circuits for the logical $\ket{0}_L^{\otimes k}$ state constructed with our methods for various $d=3$, $d=5$ and $d=7$ CSS codes under circuit-level depolarizing noise ($p_\mathrm{idle}=p/100$) using a LUT decoder.}
  \label{fig:results_all}
\end{figure*}
\begin{figure*}[t]
 \captionsetup{
    justification=raggedright,
}
  \centering
  \begin{subfigure}[t]{.8\linewidth}
    \centering
    \includegraphics[width=\linewidth]{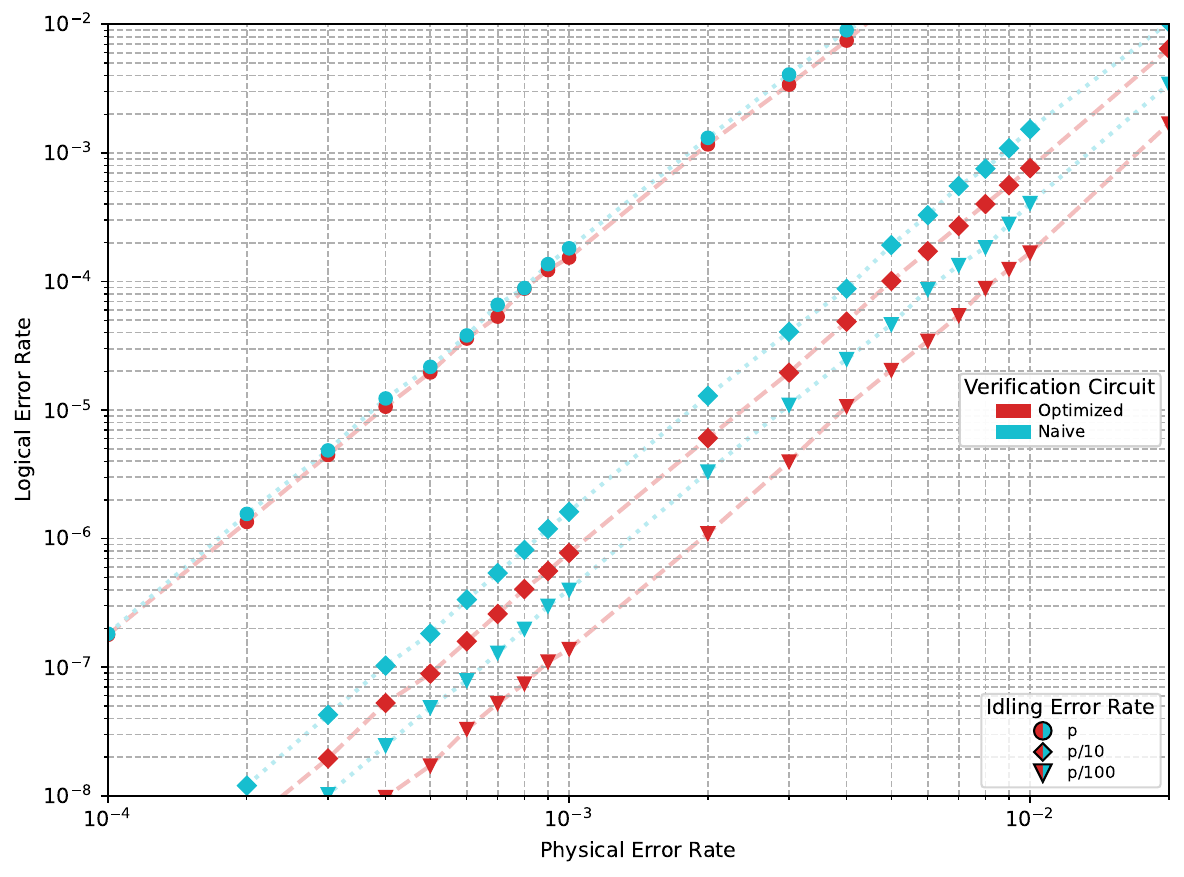}
  \end{subfigure}

  \begin{subfigure}[t]{.8\linewidth}
    \centering
    \includegraphics[width=\linewidth]{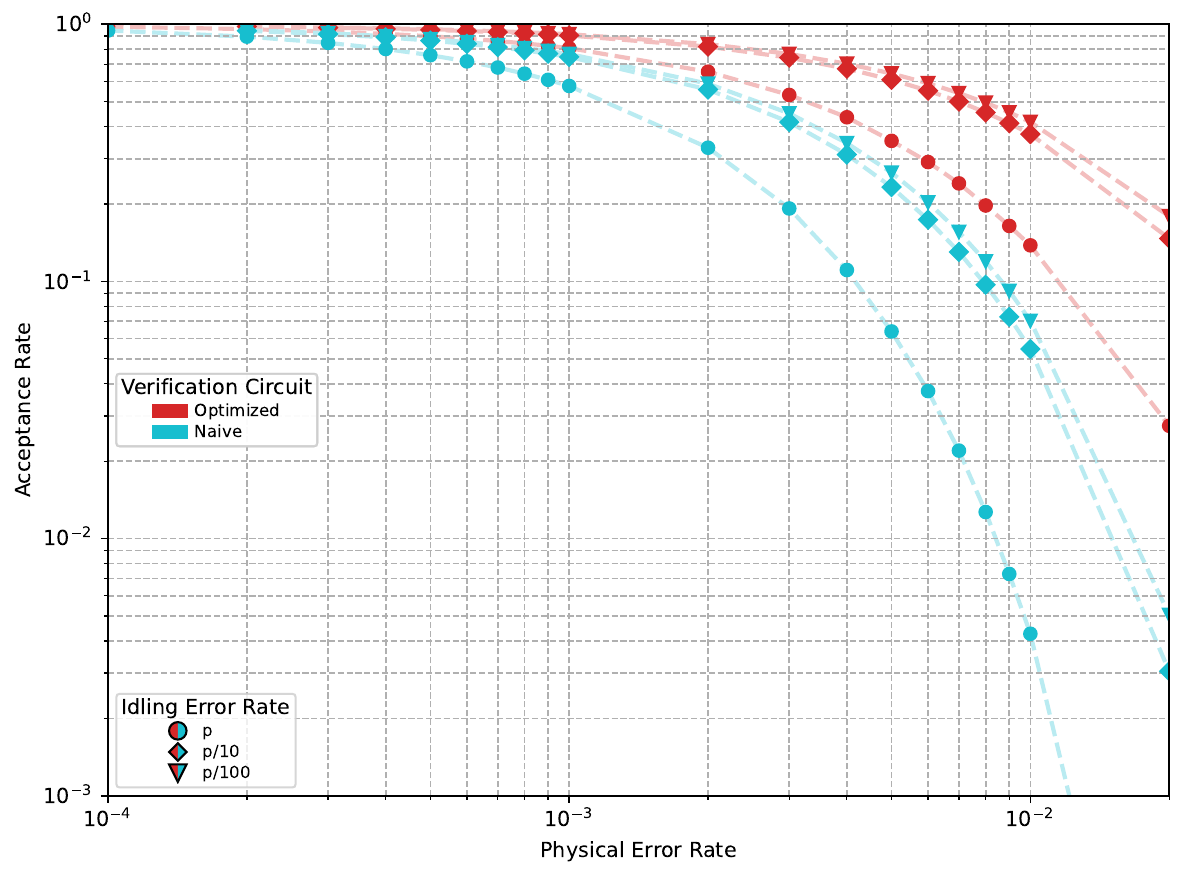}
  \end{subfigure}
  \caption{
  Logical error rate and acceptance rate for non-deterministic fault-tolerant state preparation circuits for the logical $\ket{0}_L$ state of the $\llbracket 17, 1, 5\rrbracket$ color code constructed with our methods under circuit-level depolarizing noise with \emph{parallel} gate execution using a LUT decoder.}
  \label{fig:results_d5_parallel}
\end{figure*}
\begin{figure*}[t]
 \captionsetup{
    justification=raggedright,
}
  \centering
  \begin{subfigure}[t]{.8\linewidth}
    \centering
    \includegraphics[width=\linewidth]{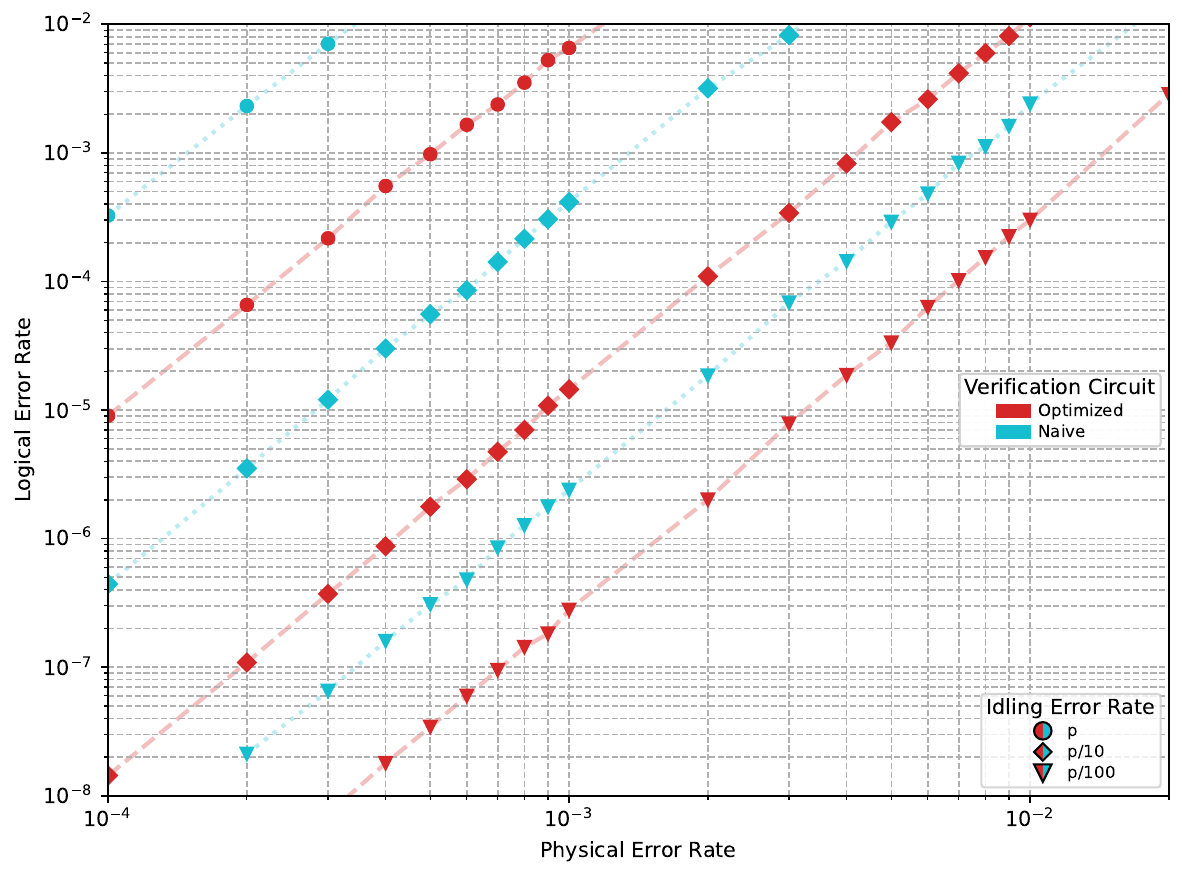}
  \end{subfigure}

  \begin{subfigure}[t]{.8\linewidth}
    \centering
    \includegraphics[width=\linewidth]{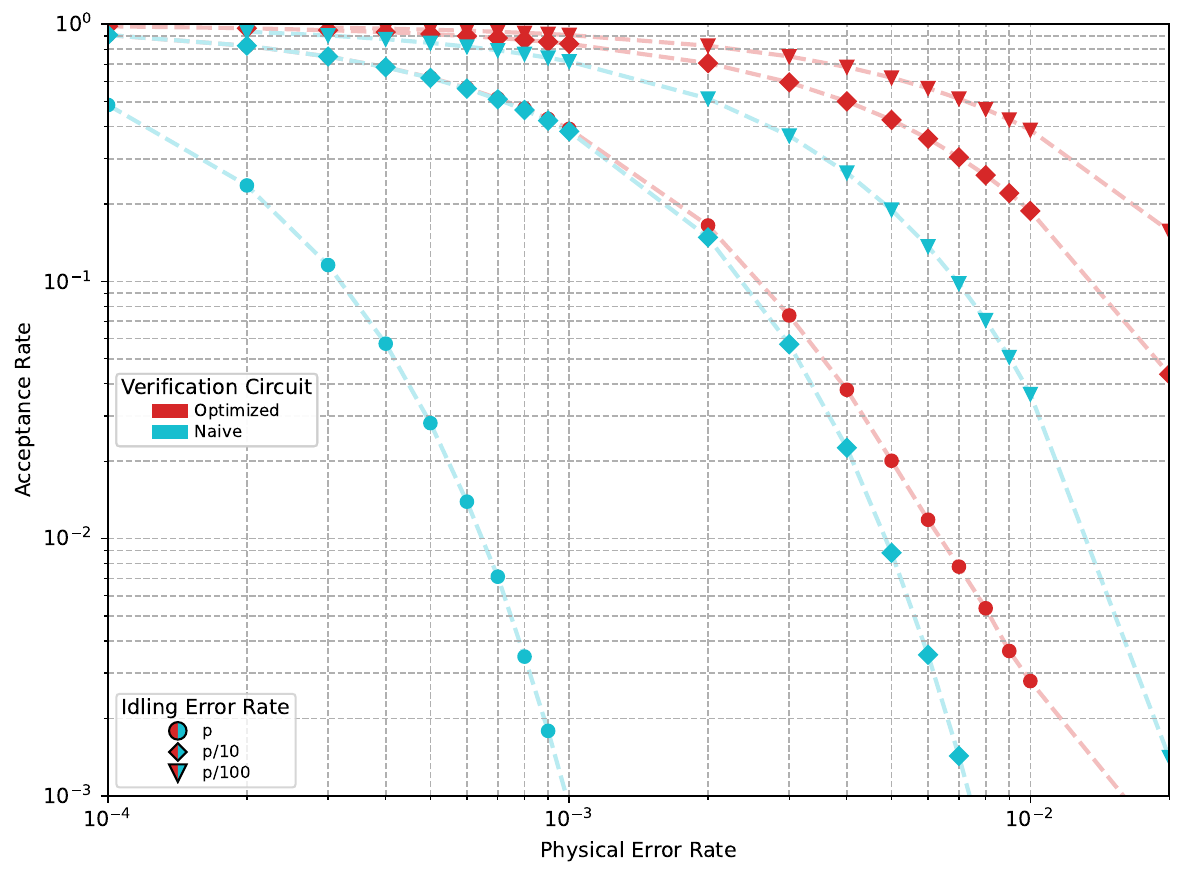}
  \end{subfigure}
  \caption{
    Logical error rate and acceptance rate for non-deterministic fault-tolerant state preparation circuits for the logical $\ket{0}_L$ state of the $\llbracket 17, 1, 5\rrbracket$ color code constructed with our methods under circuit-level depolarizing noise with \emph{sequential} gate execution using a LUT decoder.}
  \label{fig:results_d5_sequential}
\end{figure*}

The results of the simulations with parallel gates and $p_\mathrm{idle}=p/100$ are shown in~\Cref{fig:results_all}. 
For all codes, we find that the observed error suppression is in agreement with the expected error rate relative to the distance $d$ of the code~\footnote{Note that, as pointed out e.g., in Ref.~\cite{m.p.dasilvaDemonstrationLogicalQubits2024}, if one operates the quantum circuit for the $\llbracket 12,2,4 \rrbracket$ carbon code in an error detecting regime, i.e., each run with any error (on data qubits) triggering a stabilizer is discarded, all two-error events will be sorted out, and one obtains a $\sim p^3$ scaling for the logical error rate. In our simulations, we focus only on single errors in the preparation of the carbon code states.}. 
All simulated distance $d=3$ codes exhibit a \mbox{$p_L\sim O(p^2)$} scaling of the logical error rate, the distance $d=5$ codes exhibit a $p_L\sim O(p^3)$ scaling and the $\llbracket 31, 1, 7 \rrbracket$ color code exhibits a $p_L\sim O(p^4)$ scaling. 
This shows that the circuits produced with our methods indeed generalize the post-selection scheme previously known for distance $d=3$ codes~\cite{buttFaultTolerantCodeSwitchingProtocols2024,postlerDemonstrationFaulttolerantUniversal2022, gotoMinimizingResourceOverheads2016, zenQuantumCircuitDiscovery2024} to higher distances $d>3$. 
Moreover, though the acceptance rates for larger codes are drastically worse for high physical error rates, this difference becomes less pronounced for the $d=5$ codes, where all verification circuits accept the prepared state in $>80\%$ of the cases at physical error rates of $10^{-3}$.

To illustrate the performance of distance $d=5$ state preparation more clearly, we compare the optimized state preparation circuits with a \enquote{naive} (canonical) verification circuits, where all stabilizers are measured twice.
In addition, to visualize the impact of gate parallelism, we simulate the circuits with parallel gates and sequential gates.
The respective results for the $\llbracket 17,1,5\rrbracket$ color code can be seen in \Cref{fig:results_d5_parallel} and \Cref{fig:results_d5_sequential}.
Similar plots for the $\llbracket 19,1,5\rrbracket$ color code and the $\llbracket 25,1,5\rrbracket$ rotated surface code are shown in~\Cref{sec:appendix-plots}.

These results show that the difference between logical error rates of the optimized versus the naive scheme decreases with higher idling noise.
This is reasonable since, depth-wise, there is not much difference between the two circuits.
In fact, the naive scheme has fewer idling qubits if CNOTs are executed in parallel and, therefore, if the idling noise is in the same order as the two-qubit gate noise, this behavior is expected.
The story is a bit different when looking at acceptance rates.
With $p>10^{-3}$ the difference between the acceptance rates of the naive and optimized schemes is quite stark.
For example, at $p=4\cdot 10^{-3}$ and $p_\mathrm{idle}=p/100$, the optimized circuit has an acceptance rate of about $70\%$, whereas the naive verification circuit gives an acceptance rate of less than $20\%$.

When gates are executed sequentially, both circuits are subject to about the same amount of idling noise.
In this case, the optimized circuit outperforms the naive circuit by orders of magnitude even for comparatively low idling noise.

These results show that optimized state preparation schemes are particularly suited for quantum computing platforms with a lower degree of parallelism. For architectures like neutral atoms~\cite{everedHighfidelityParallelEntangling2023,bluvsteinLogicalQuantumProcessor2024} where two-qubit gates can be performed globally, the difference is less stark, assuming idling errors are manageable. For these systems, alternative schemes with a high degree of parallelism, like in the preparation of the Golay code~\cite{paetznickFaulttolerantAncillaPreparation2013}, could be exploited.

\section{Conclusion}
\label{sec:conclusion}
In this work, we have introduced automated methods to synthesize (near-)optimal non-deterministic fault-tolerant state preparation circuits for logical computational basis states of any CSS code. 
We have shown the efficacy of the proposed methods by synthesizing fault-tolerant state preparation circuits for various small CSS codes---relevant for near-term devices---and demonstrated that these circuits have the desired fault-tolerance properties through numerical simulations. 
Even though we have framed the verification circuit synthesis problem in the context of state preparation of certain logical basis states, our verification circuit synthesis approach is more general. It can be used to verify other circuits as well.

In this work, we have focused on constraints on specific stabilizer measurements, but generally, in fault-tolerant circuit design, we are more concerned with constraints between different measurements. Looking forward, an interesting next step would be to apply methods used in this work to the design of fault-tolerant circuits from the perspective of detector-error models as in Ref.~\cite{derksDesigningFaulttolerantCircuits2024}.

As discussed in~\Cref{sec:fully-fault-tolerant}, our methods do not guarantee a globally optimal circuit but rather give locally optimal subcircuits for the state and verification circuit synthesis problems. 
The natural next step would be to try to unify this into one synthesis problem. 
Given that the fault set of the non-fault-tolerant state preparation circuit has a huge impact on the verification circuit, we suspect there is a large room for improvement, especially when scaling to even higher distances. Alternatively, one could use these circuits in a concatenated scheme. Whether this can be extended to concatenation in a trivial matter or whether further measurements are needed to verify error propagation through the transversal CNOTs is left for future work.

In our work, we have assumed all-to-all connectivity between all data and ancilla qubits. 
Many quantum computing platforms do not provide such lenient locality constraints; however, further steps must be taken to ensure that non-deterministic fault-tolerant state preparation circuits can be embedded into such architectures without sacrificing fault-tolerance.

All in all, this work is an important step towards providing automated design tools for fault-tolerant quantum computing research. 
The optimized circuits generated by our methods might play a role in further near-term experimental demonstrations of fault-tolerant quantum computing. 

\section{Used Software}
\label{sec:used-software}

The proposed methods have been implemented using the Python packages Stim~\cite{gidneyStimFastStabilizer2021}, NumPy~\cite{harrisArrayProgrammingNumPy2020}, Z3~\cite{leonardodemouraZ3EfficientSMT2008}, Qiskit~\cite{qiskitcontributorsQiskitOpensourceFramework2023}, LDPC~\cite{roffeLDPCPythonTools2022} and Multiprocess~\cite{mckernsBuildingFrameworkPredictive2012,michaelmckernsPathosFrameworkHeterogeneous2010}. For the simulations we have furthermore used GNU parallel~\cite{tangeGNUParallel20182018}. The plots have been generated with Matplotlib~\cite{hunterMatplotlib2DGraphics2007}.

\acknowledgements{
The authors would like to thank Timo Hillmann for valuable comments on an initial version of the manuscript.

T.P., L.S., L.B., and R.W. acknowledge funding from the European Research Council (ERC) under the European Union’s Horizon 2020 research and innovation program (grant agreement No.\ 101001318) and Millenion, grant agreement No.\ 101114305). 
This work was part of the Munich Quantum Valley, which is supported by the Bavarian state government with funds from the Hightech Agenda Bayern Plus, and has been supported by the BMWK on the basis of a decision by the German Bundestag through project QuaST, as well as by 
the BMK, BMDW, and the State of Upper Austria in the 
frame of the COMET program (managed by the FFG).

Furthermore, MM gratefully acknowledges support by the European Union’s Horizon Europe research and innovation programme under Grant Agreement No. 101114305 (“MILLENION-SGA1” EU Project) and the ERC Starting Grant QNets through Grant No. 804247. This research is also part of the Munich Quantum Valley (K-8), which is supported by the Bavarian state government with funds from the Hightech Agenda Bayern Plus. He additionally acknowledges support by the BMBF project MUNIQC-ATOMS (Grant No. 13N16070), by the Deutsche Forschungsgemeinschaft (DFG, German Research Foundation) under Germany’s Excellence Strategy “Cluster of Excellence Matter and Light for Quantum Computing (ML4Q) EXC 2004/1” 390534769.

\section*{Data Availability}
\label{par:data}

The data that support the findings of this article are openly available~\cite{berentMQTQECC2025, pehamDataAutomatedSynthesis2025}.

\bibliography{lit_header,zotero}

\appendix

\clearpage
\onecolumngrid
\section{Further Plots for $d=5$ Codes}\label{sec:appendix-plots}
\vspace*{-5mm}
\begin{figure*}[h!]
 \captionsetup{
    justification=raggedright,
}
  \centering
  \begin{subfigure}[t]{.75\linewidth}
    \centering
    \includegraphics[width=\linewidth]{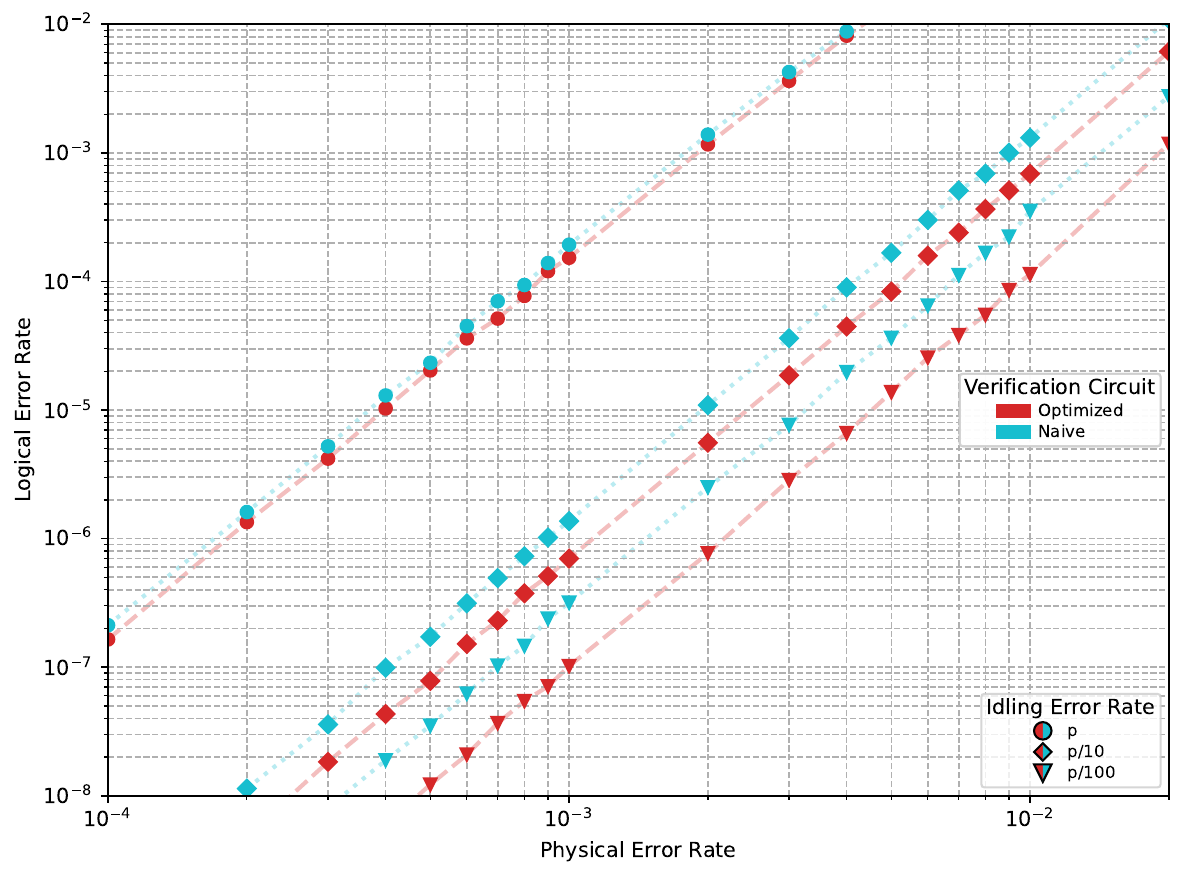}
  \end{subfigure}

  \begin{subfigure}[t]{.75\linewidth}
    \centering
    \includegraphics[width=\linewidth]{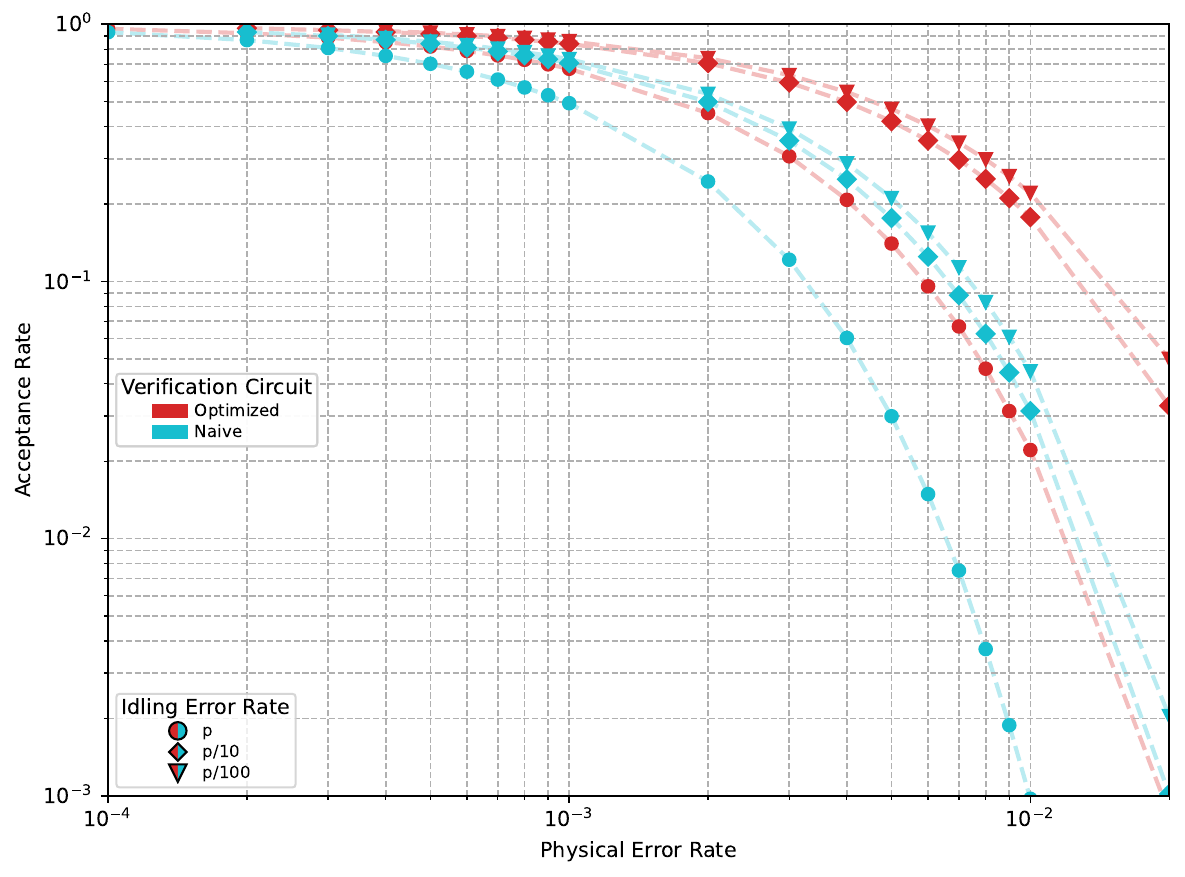}
  \end{subfigure}
  \caption{
  Logical error rate and acceptance rate for non-deterministic fault-tolerant state preparation circuits for the logical $\ket{0}_L$ state of the $\llbracket 19, 1, 5\rrbracket$ color code constructed with our methods under circuit-level depolarizing noise with \emph{parallel} gate execution using a LUT decoder.}
  \label{fig:results_cc_666_5_parallel}
\end{figure*}

\begin{figure*}[h]
 \captionsetup{
    justification=raggedright,
}
  \centering
  \begin{subfigure}[t]{.8\linewidth}
    \centering
    \includegraphics[width=\linewidth]{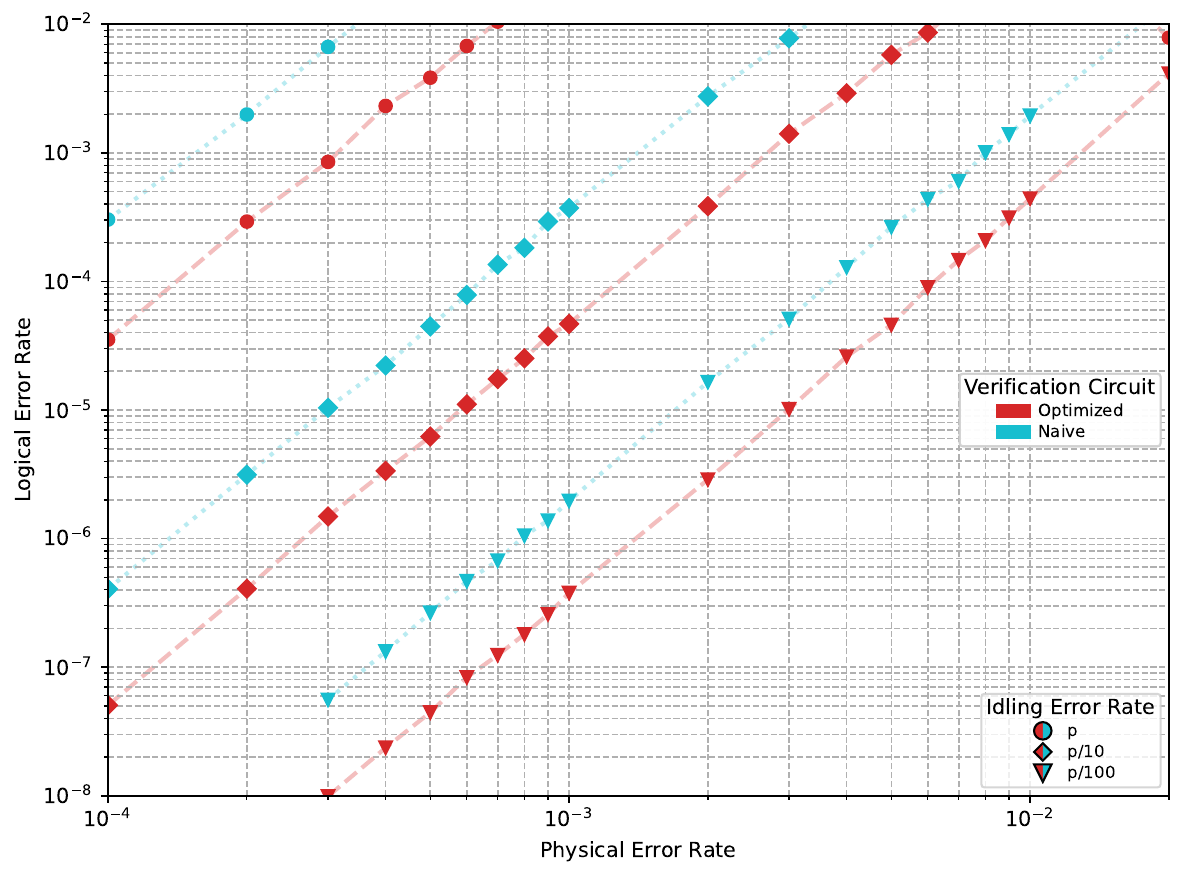}
  \end{subfigure}

  \begin{subfigure}[t]{.8\linewidth}
    \centering
    \includegraphics[width=\linewidth]{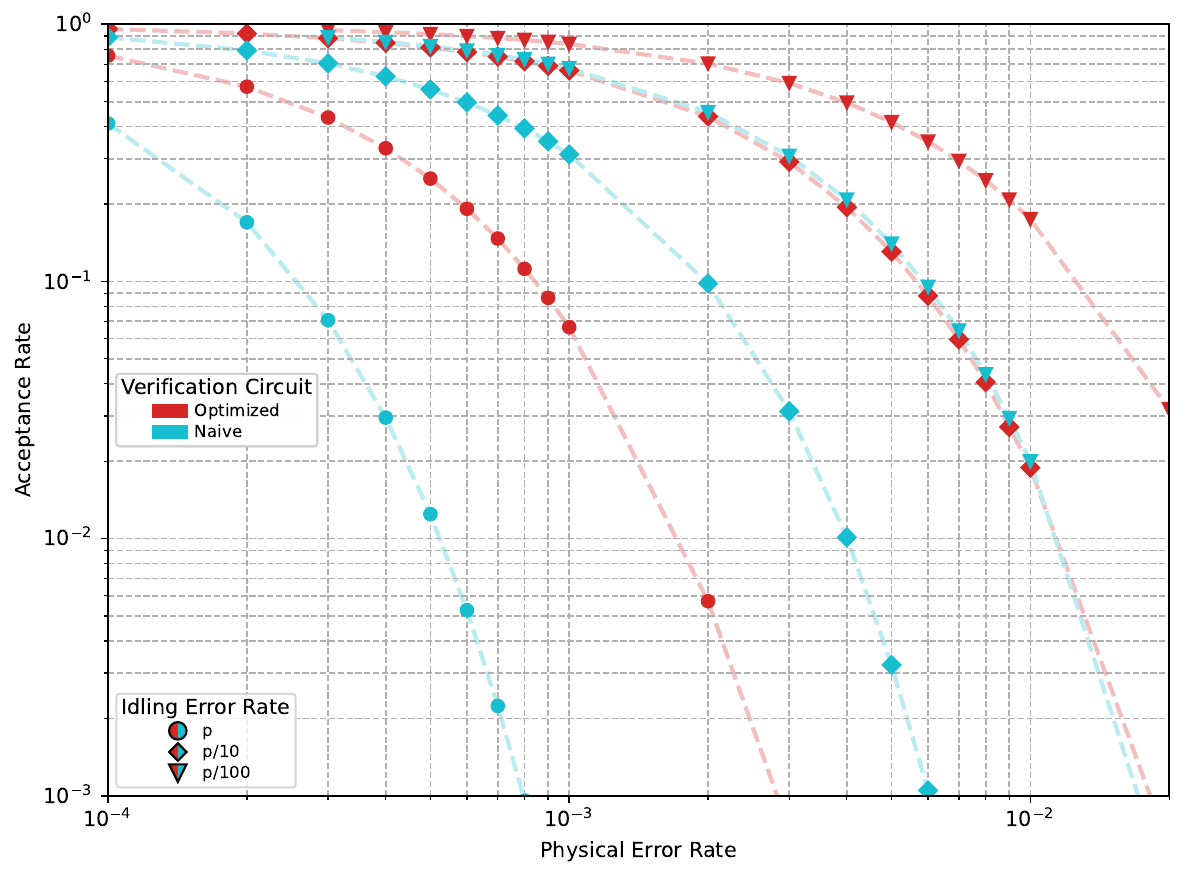}
  \end{subfigure}
  \caption{
    Logical error rate and acceptance rate for non-deterministic fault-tolerant state preparation circuits for the logical $\ket{0}_L$ state of the $\llbracket 19, 1, 5\rrbracket$ color code constructed with our methods under circuit-level depolarizing noise with \emph{sequential} gate execution using a LUT decoder.}
  \label{fig:results_cc_666_5_sequential}
\end{figure*}

\begin{figure*}[h]
 \captionsetup{
    justification=raggedright,
}
  \centering
  \begin{subfigure}[t]{.8\linewidth}
    \centering
    \includegraphics[width=\linewidth]{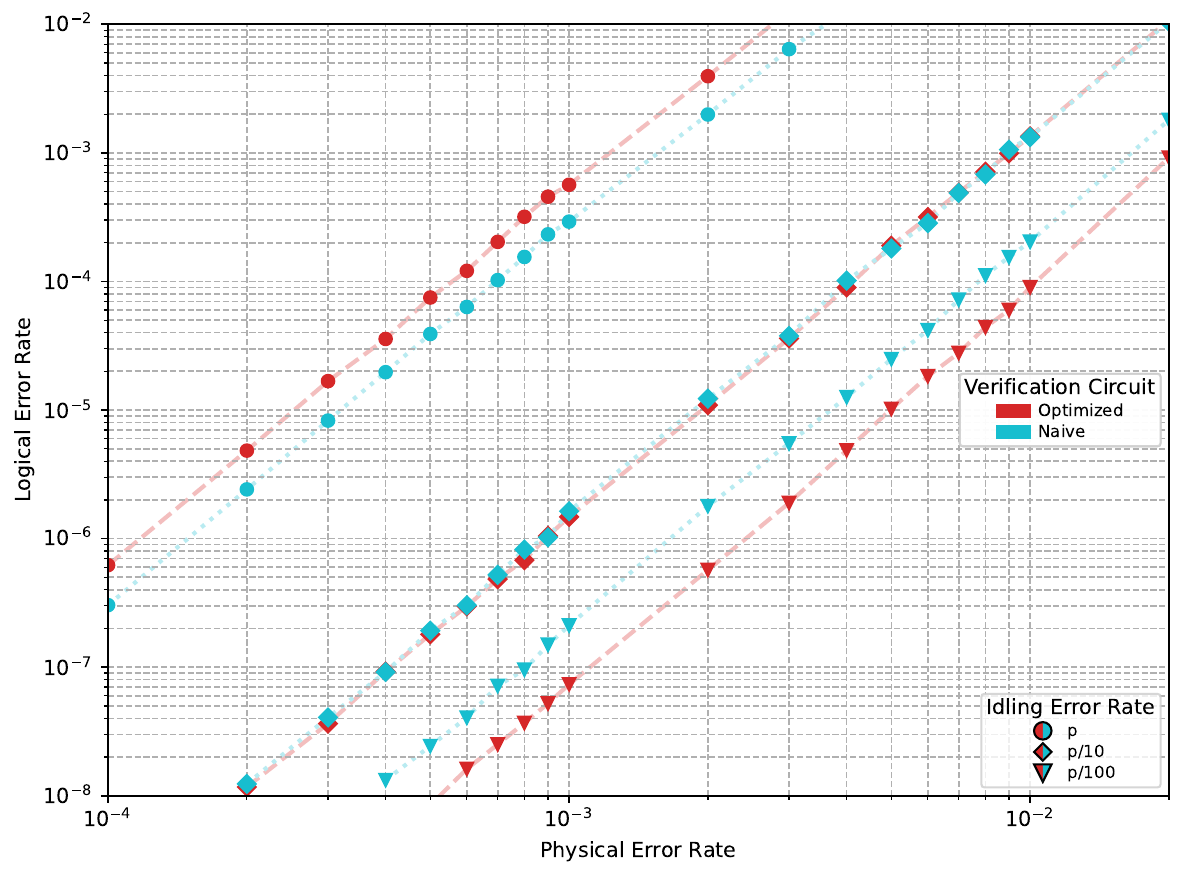}
  \end{subfigure}

  \begin{subfigure}[t]{.8\linewidth}
    \centering
    \includegraphics[width=\linewidth]{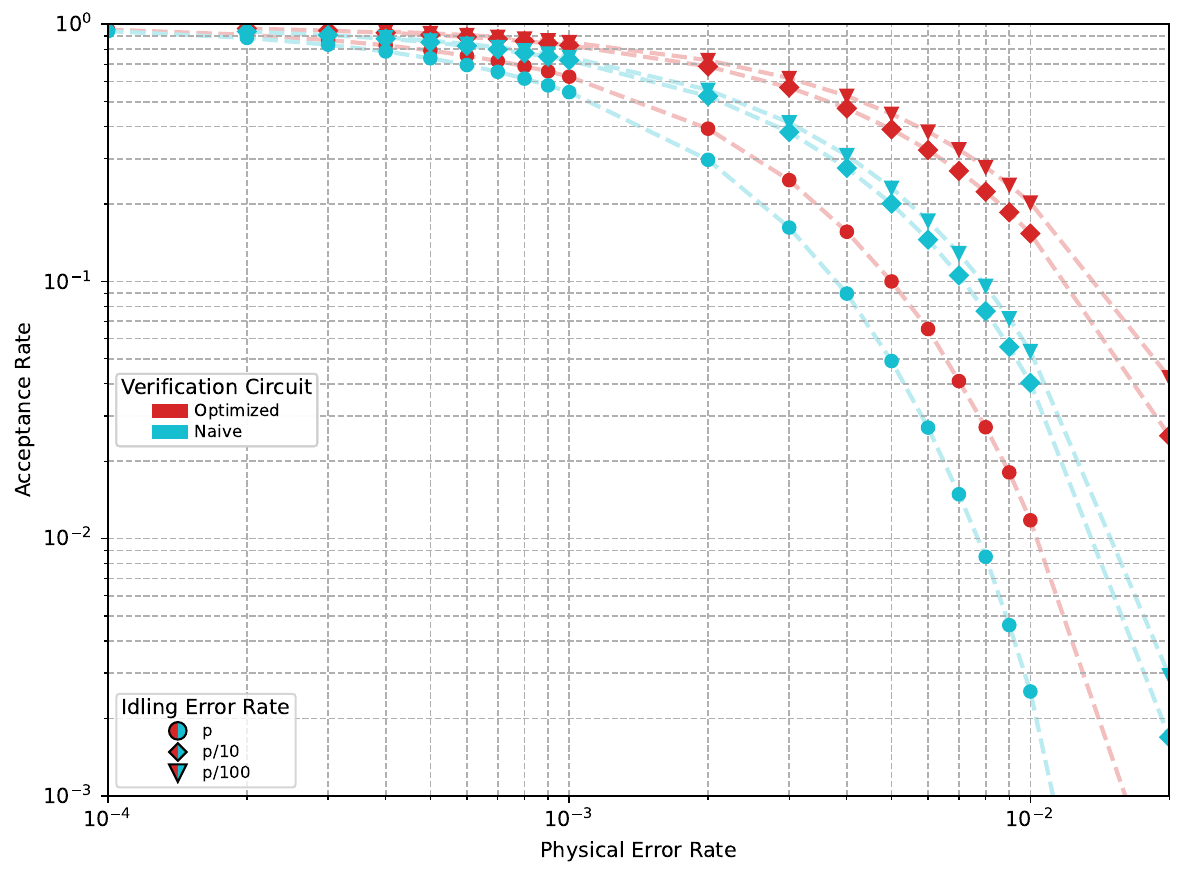}
  \end{subfigure}
  \caption{
    Logical error rate and acceptance rate for non-deterministic fault-tolerant state preparation circuits for the logical $\ket{0}_L$ state of the $\llbracket 25, 1, 5\rrbracket$ rotated surface  code constructed with our methods under circuit-level depolarizing noise with \emph{parallel} gate execution using a LUT decoder.}
  \label{fig:results_surface_5_parallel}
\end{figure*}

\begin{figure*}[h]
 \captionsetup{
    justification=raggedright,
}
  \centering
  \begin{subfigure}[t]{.8\linewidth}
    \centering
    \includegraphics[width=\linewidth]{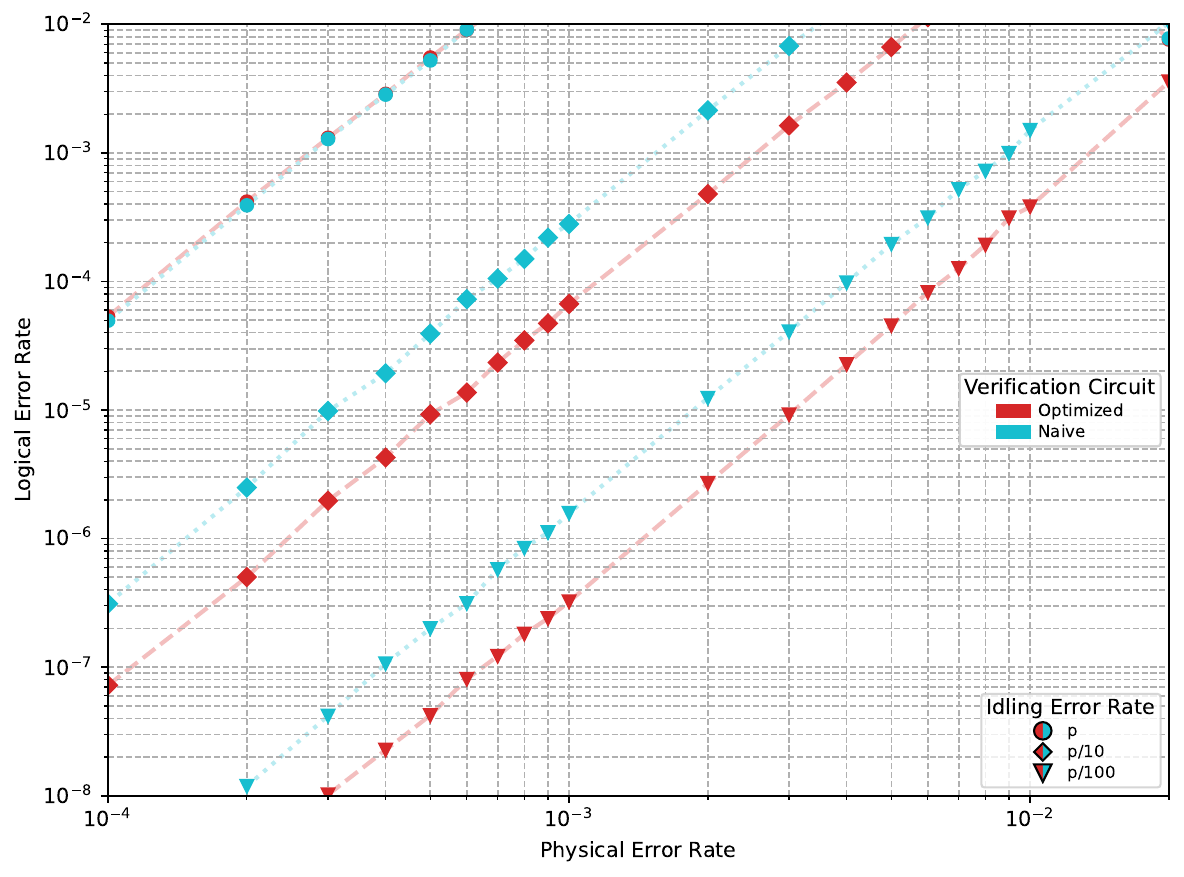}
  \end{subfigure}

  \begin{subfigure}[t]{.8\linewidth}
    \centering
    \includegraphics[width=\linewidth]{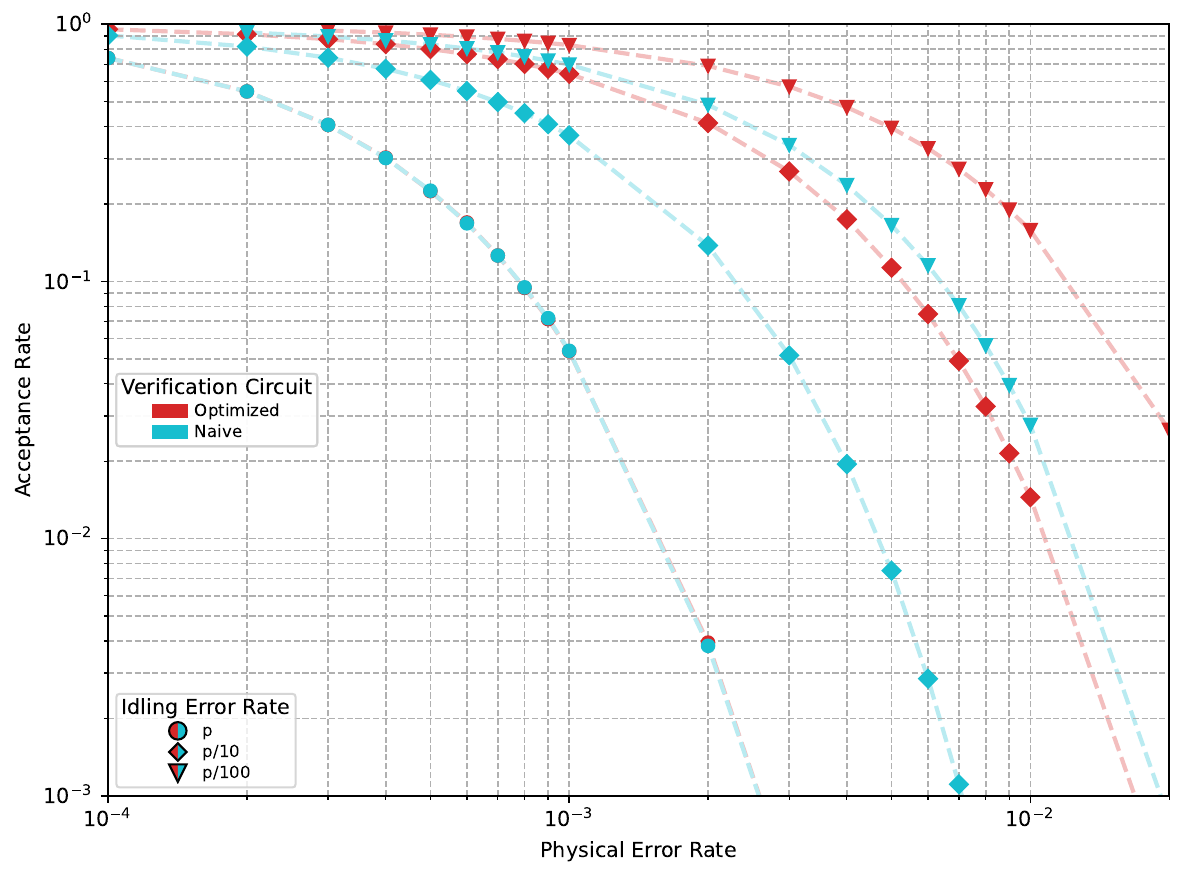}
  \end{subfigure}
  \caption{
    Logical error rate and acceptance rate for non-deterministic fault-tolerant state preparation circuits for the logical $\ket{0}_L$ state of the $\llbracket 25, 1, 5\rrbracket$ rotated surface code constructed with our methods under circuit-level depolarizing noise with \emph{sequential} gate execution using a LUT decoder.}
  \label{fig:results_surface_5_sequential}
\end{figure*}
\clearpage
\newlength{\slength}
\settowidth{\slength}{$+$}
\section{Circuits}\label{sec:appendix-circuits}

\begin{figure*}[htb!]
\begin{subfigure}[h!]{\linewidth}
  \centering
  \begin{tikzpicture}
    \begin{yquant}
      qubit {$\ket{+}$} q;
      qubit {$\ket{\makebox[\slength][c]{0}}$} q[+7];
      qubit {$\ket{+}$} q[+1];
      qubit {$\ket{\makebox[\slength][c]{0}}$} q[+1];
      qubit {$\ket{+}$} q[+1];
      qubit {$\ket{\makebox[\slength][c]{0}}$} q[+3];
      qubit {$\ket{+}$} q[+1];
      
      cnot q[7] | q[14];
      cnot q[3] | q[10];
      cnot q[1] | q[8];
      cnot q[12] | q[14];
      cnot q[9] | q[10];
      cnot q[5] | q[7];
      cnot q[2] | q[3];
      cnot q[1] | q[0];
      cnot q[13] | q[14];
      cnot q[11] | q[12];
      cnot q[9] | q[8];
      cnot q[6] | q[7];
      cnot q[4] | q[5];
      cnot q[3] | q[0];
      cnot q[2] | q[1];
      cnot q[13] | q[10];
      cnot q[12] | q[8];
      cnot q[7] | q[0];
      cnot q[11] | q[9];
      cnot q[6] | q[3];
      cnot q[4] | q[1];
      cnot q[5] | q[2];

      [after=q]
      qubit {$\ket{+}$} verification_flag[1];
      cnot verification_flag[0] | q[0];
      cnot verification_flag[0] | q[5];
      cnot verification_flag[0] | q[11];
      measure verification_flag[0];
    \end{yquant}
  \end{tikzpicture}
  \caption{Non-deterministic fault-tolerant $\ket{0}_L$ of the $\llbracket 15,1,3 \rrbracket$ code}
  \label{fig:tetrahedral_zero}
\end{subfigure}
\bigskip

\begin{subfigure}[h]{\linewidth}
  \centering
  \resizebox{\linewidth}{!}{
  \begin{tikzpicture}
    \begin{yquant}
      qubit {$\ket{\makebox[\slength][c]{0}}$} q[2];
      qubit {$\ket{+}$} q[+1];
      qubit {$\ket{\makebox[\slength][c]{0}}$} q[+1];
      qubit {$\ket{+}$} q[+1];
      qubit {$\ket{\makebox[\slength][c]{0}}$} q[+1];
      qubit {$\ket{+}$} q[+1];
      qubit {$\ket{\makebox[\slength][c]{0}}$} q[+3];
      qubit {$\ket{+}$} q[+1];
      qubit {$\ket{\makebox[\slength][c]{0}}$} q[+2];
      qubit {$\ket{+}$} q[+1];
      qubit {$\ket{\makebox[\slength][c]{0}}$} q[+1];

cnot q[11]|q[13];
cnot q[12]|q[11];
cnot q[5]|q[6];
cnot q[3]|q[10];
cnot q[1]|q[4];
cnot q[9]|q[11];
cnot q[8]|q[12];
cnot q[7]|q[4];
cnot q[5]|q[3];
cnot q[1]|q[2];
cnot q[14]|q[7];
cnot q[13]|q[6];
cnot q[9]|q[2];
cnot q[8]|q[1];
cnot q[0]|q[3];
cnot q[4]|q[5];
cnot q[14]|q[13];
cnot q[10]|q[9];
cnot q[0]|q[1];
cnot q[12]|q[4];
cnot q[7]|q[6];
cnot q[11]|q[5];
cnot q[3]|q[2];

[after=q]
qubit {$\ket{+}$} x_anc[1];
cnot q[1]|x_anc[0];
cnot q[3]|x_anc[0];
cnot q[5]|x_anc[0];
cnot q[7]|x_anc[0];
cnot q[9]|x_anc[0];
cnot q[12]|x_anc[0];
cnot q[13]|x_anc[0];

measure {$X$} x_anc[0];

[after=q]
qubit  {$\ket{0}$} z_anc[2];

cnot z_anc[0]|q[0];

[after=z_anc]
qubit {$\ket{+}$} a15[1];
cnot z_anc[0]|a15[0];
cnot z_anc[0]|q[5];
cnot z_anc[0]|q[9];
cnot z_anc[0]|a15[0];
cnot z_anc[0]|q[14];
measure z_anc[0];
cnot z_anc[1]|q[3];
[after=z_anc]
qubit {$\ket{+}$} a16[1];
cnot z_anc[1]|a16[0];
cnot z_anc[1]|q[7];
cnot z_anc[1]|q[9];
cnot z_anc[1]|a16[0];
cnot z_anc[1]|q[12];

measure {$X$} a16[0];
measure {$X$} a15[0];
measure z_anc[1];
    \end{yquant}
  \end{tikzpicture}}
  \caption{Non-deterministic fault-tolerant $\ket{+}_L$ of the $\llbracket 15, 1, 3 \rrbracket$ code.}
  \label{fig:tetrahedral_plus}

\end{subfigure}
  \caption{Fault-tolerant state preparation circuits for the $\llbracket 15, 1, 3 \rrbracket$ \enquote{tetrahedral} code.}\label{fig:tetrahedral_circuits}
\end{figure*}

\begin{figure*}[htb!]
  \centering
  \resizebox{\textwidth}{!}{
  \begin{tikzpicture}
    \begin{yquant}

      qubit {$\ket{\makebox[\slength][c]{0}}$} q[1];
      qubit {$\ket{+}$} q[+2];
      qubit {$\ket{\makebox[\slength][c]{0}}$} q[+1];
      qubit {$\ket{+}$} q[+1];
      qubit {$\ket{\makebox[\slength][c]{0}}$} q[+1];
      qubit {$\ket{+}$} q[+1];
      qubit {$\ket{\makebox[\slength][c]{0}}$} q[+1];
      qubit {$\ket{+}$} q[+1];
      qubit {$\ket{\makebox[\slength][c]{0}}$} q[+3];
      cnot q[7] | q[1];
      cnot q[9] | q[6];
      cnot q[5] | q[2];
      cnot q[0] | q[7];
      cnot q[6] | q[7];
      cnot q[10] | q[4];
      cnot q[11] | q[5];
      cnot q[3] | q[9];
      cnot q[1] | q[10];
      cnot q[5] | q[8];
      cnot q[2] | q[7];
      cnot q[4] | q[8];
      cnot q[9] | q[8];
      cnot q[8] | q[2];
      cnot q[0] | q[9];
      cnot q[7] | q[4];

      [after=q]
      qubit {$\ket{\makebox[\slength][c]{0}}$} z_anc[1];
      cnot z_anc[0] | q[0];
      cnot z_anc[0] | q[2];
      cnot z_anc[0] | q[3];
      cnot z_anc[0] | q[4];
      cnot z_anc[0] | q[10];
      cnot z_anc[0] | q[11];
      measure z_anc[0];

      [after=z_anc]
      qubit {$\ket{+}$} x_anc[1];
      [after=z_anc]
      qubit {$\ket{\makebox[\slength][c]{0}}$} a62[1];
      cnot q[0] | x_anc[0];
      cnot a62[0] | x_anc[0];
      cnot q[1] | x_anc[0];
      cnot q[5] | x_anc[0];
      cnot q[6] | x_anc[0];
      cnot q[7] | x_anc[0];
      cnot a62[0] | x_anc[0];
      measure a62[0];
      cnot q[11] | x_anc[0];
      measure {$X$} x_anc[0];
    \end{yquant}
  \end{tikzpicture}}
  \caption{Non-deterministic fault-tolerant $\ket{0}_L^2$ of the $\llbracket 12, 2, 4 \rrbracket$ \enquote{Carbon} code.}
  \label{fig:carbon}
\end{figure*}

\begin{figure*}[h]
  \centering
  \resizebox{\textwidth}{!}{
  \begin{tikzpicture}
    \begin{yquant}
      qubit {$\ket{+}$} q[2];
      qubit {$\ket{\makebox[\slength][c]{0}}$} q[+9];
      qubit {$\ket{+}$} q[+2];
      qubit {$\ket{\makebox[\slength][c]{0}}$} q[+2];

      cnot q[10] | q[0];
      cnot q[0] | q[1];
      cnot q[7] | q[1];
      cnot q[3] | q[12];
      cnot q[6] | q[10];
      cnot q[8] | q[1];
      cnot q[4] | q[11];
      cnot q[11] | q[3];
      cnot q[13] | q[11];
      cnot q[9] | q[8];
      cnot q[14] | q[0];
      cnot q[10] | q[12];
      cnot q[3] | q[1];
      cnot q[5] | q[6];
      cnot q[0] | q[11];
      cnot q[6] | q[0];
      cnot q[7] | q[0];
      cnot q[1] | q[10];
      cnot q[11] | q[1];
      cnot q[9] | q[4];
      cnot q[2] | q[4];
      cnot q[4] | q[5];

      [after=q]
      qubit {$\ket{0}$} z_anc[2];

      cnot z_anc[0] | q[0];
      cnot z_anc[0] | q[9];
      cnot z_anc[0] | q[10];
      cnot z_anc[1] | q[4];
      cnot z_anc[1] | q[8];
      cnot z_anc[1] | q[11];
      measure z_anc[0];
      measure z_anc[1];
    \end{yquant}
  \end{tikzpicture}}
  \caption{Non-deterministic fault-tolerant $\ket{0}_L^7$ of the $\llbracket 15, 7, 3 \rrbracket$ Hamming code. }
  \label{fig:hamming}
\end{figure*}

\onecolumngrid

\begin{sidewaysfigure*}
  \centering
  \resizebox{\textwidth}{!}{
  \begin{tikzpicture}
    \begin{yquant}

      qubit {$\ket{+}$} q[6];
      qubit {$\ket{\makebox[\slength][c]{0}}$} q[+3];
      qubit {$\ket{+}$} q[+2];
      qubit {$\ket{\makebox[\slength][c]{0}}$} q[+7];
      
      cnot q[8] | q[4];
      cnot q[15] | q[10];
      cnot q[6] | q[4];
      cnot q[15] | q[9];
      cnot q[12] | q[5];
      cnot q[7] | q[4];
      cnot q[14] | q[3];
      cnot q[6] | q[2];
      cnot q[16] | q[14];
      cnot q[13] | q[10];
      cnot q[11] | q[9];
      cnot q[5] | q[8];
      cnot q[7] | q[3];
      cnot q[4] | q[1];
      cnot q[2] | q[0];
      cnot q[9] | q[12];
      cnot q[13] | q[1];
      cnot q[3] | q[0];
      cnot q[8] | q[15];
      cnot q[16] | q[6];
      cnot q[11] | q[5];
      cnot q[10] | q[4];
      cnot q[14] | q[2];

      [after=q]
      qubit {$\ket{\makebox[\slength][c]{0}}$} z_anc[5];
      cnot z_anc[0] | q[1];
      cnot z_anc[0] | q[5];
      cnot z_anc[0] | q[9];
      cnot z_anc[0] | q[10];
      cnot z_anc[0] | q[15];
      cnot z_anc[1] | q[0];
      cnot z_anc[1] | q[2];
      cnot z_anc[1] | q[7];
      cnot z_anc[1] | q[16];
      cnot z_anc[2] | q[0];
      cnot z_anc[2] | q[2];
      cnot z_anc[2] | q[7];
      cnot z_anc[2] | q[16];
      cnot z_anc[3] | q[1];
      cnot z_anc[3] | q[3];
      cnot z_anc[3] | q[4];
      cnot z_anc[3] | q[6];
      cnot z_anc[3] | q[14];
      cnot z_anc[4] | q[1];
      cnot z_anc[4] | q[4];
      cnot z_anc[4] | q[9];
      cnot z_anc[4] | q[11];

      [after=z_anc]
      qubit {$\ket{\makebox[\slength][c]{+}}$} x_anc[3];
      [after=z_anc]
      qubit {$\ket{\makebox[\slength][c]{0}}$} a12[3];
      [after=z_anc]
      qubit {$\ket{\makebox[\slength][c]{0}}$} a13[1];
      [after=z_anc]
      qubit {$\ket{\makebox[\slength][c]{0}}$} a14[1];

      cnot q[0] | x_anc[0];
      cnot a12[0] | x_anc[0];
      cnot q[1] | x_anc[0];
      cnot a12[1] | x_anc[0];
      cnot q[8] | x_anc[0];
      cnot q[10] | x_anc[0];
      cnot a12[2] | x_anc[0];
      cnot q[11] | x_anc[0];
      cnot q[12] | x_anc[0];
      cnot a12[0] | x_anc[0];

      cnot q[14] | x_anc[0];
      cnot a12[2] | x_anc[0];

      cnot a12[1] | x_anc[0];
      
      cnot q[16] | x_anc[0];

      cnot q[5] | x_anc[1];
      cnot a13[0] | x_anc[1];
      cnot q[10] | x_anc[1];
      cnot q[12] | x_anc[1];
      cnot a13[0] | x_anc[1];
      
      cnot q[13] | x_anc[1];

      cnot q[1] | x_anc[2];
      cnot a14[0] | x_anc[2];
      cnot q[4] | x_anc[2];
      cnot q[10] | x_anc[2];
      cnot a14[0] | x_anc[2];
      
      cnot q[13] | x_anc[2];

      measure z_anc[0];
      measure z_anc[1];
      measure z_anc[2];
      measure z_anc[3];
      measure z_anc[4];
      measure a14[0];
      measure {$X$} x_anc[0];
      measure {$X$} x_anc[2];
      measure {$X$} x_anc[1];
      measure a12;
      measure a13[0];
    \end{yquant}
  \end{tikzpicture}}
\caption{Non-deterministic fault-tolerant $\ket{0}_L$ of the $\llbracket 17, 1, 5 \rrbracket$ code.}\label{fig:ket0-17-1-5}
\end{sidewaysfigure*}

\begin{sidewaysfigure*}
  \centering
  \resizebox{\textwidth}{!}{
  \begin{tikzpicture}
    \begin{yquant}

      qubit {$\ket{\makebox[\slength][c]{0}}$} q[2];
      qubit {$\ket{+}$} q[+1];
      qubit {$\ket{\makebox[\slength][c]{0}}$} q[+2];
      qubit {$\ket{+}$} q[+1];
      qubit {$\ket{\makebox[\slength][c]{0}}$} q[+6];
      qubit {$\ket{+}$} q[+7];
 
  cnot q[9] | q[13];
  cnot q[6] | q[16];
  cnot q[3] | q[12];
  cnot q[16] | q[2];
  cnot q[13] | q[18];
  cnot q[11] | q[12];
  cnot q[7] | q[15];
  cnot q[3] | q[6];
  cnot q[4] | q[9];
  cnot q[0] | q[14];
  cnot q[12] | q[17];
  cnot q[10] | q[18];
  cnot q[8] | q[9];
  cnot q[7] | q[11];
  cnot q[3] | q[4];
  cnot q[2] | q[5];
  cnot q[1] | q[15];
  cnot q[0] | q[6];
  cnot q[18] | q[17];
  cnot q[15] | q[14];
  cnot q[9] | q[5];
  cnot q[11] | q[13];
  cnot q[10] | q[12];
  cnot q[8] | q[2];
  cnot q[6] | q[7];
  cnot q[4] | q[16];
  cnot q[1] | q[0];
  [after=q]

  qubit {$\ket{\makebox[\slength][c]{0}}$} verification_flag[6];
  cnot verification_flag[0] | q[0];
  cnot verification_flag[0] | q[5];
  cnot verification_flag[0] | q[6];
  cnot verification_flag[0] | q[9];
  cnot verification_flag[0] | q[12];
  cnot verification_flag[0] | q[13];
  cnot verification_flag[0] | q[15];
  cnot verification_flag[0] | q[18];
  cnot verification_flag[1] | q[4];
  cnot verification_flag[1] | q[5];
  cnot verification_flag[1] | q[6];
  cnot verification_flag[1] | q[7];
  cnot verification_flag[1] | q[8];
  cnot verification_flag[2] | q[0];
  cnot verification_flag[2] | q[2];
  cnot verification_flag[2] | q[3];
  cnot verification_flag[2] | q[7];
  cnot verification_flag[2] | q[8];
  cnot verification_flag[2] | q[15];
  cnot verification_flag[3] | q[1];
  cnot verification_flag[3] | q[6];
  cnot verification_flag[3] | q[11];
  cnot verification_flag[3] | q[13];
  cnot verification_flag[3] | q[14];
  cnot verification_flag[4] | q[1];
  cnot verification_flag[4] | q[2];
  cnot verification_flag[4] | q[5];
  cnot verification_flag[4] | q[15];
  cnot verification_flag[4] | q[16];
  cnot verification_flag[5] | q[4];
  cnot verification_flag[5] | q[5];
  cnot verification_flag[5] | q[9];
  cnot verification_flag[5] | q[16];

    measure verification_flag[0];
    measure verification_flag[1];
    measure verification_flag[2];
    measure verification_flag[3];
    measure verification_flag[4];
    measure verification_flag[5];
  \end{yquant}
\end{tikzpicture}}
\caption{Non-deterministic fault-tolerant $\ket{0}_L$ of the $\llbracket 19, 1, 5 \rrbracket$ color code. Here only verification for X errors is shown.}\label{fig:ket0-19-1-5}
  \label{fig:hex_cc_circuit}
\end{sidewaysfigure*}

\begin{sidewaysfigure*}
  \centering
  \resizebox{\textwidth}{!}{
  \begin{tikzpicture}
    \begin{yquant}
      qubit {$\ket{+}$} q[7];
      qubit {$\ket{\makebox[\slength][c]{0}}$} q[+1];
      qubit {$\ket{+}$} q[+5];
      qubit {$\ket{\makebox[\slength][c]{0}}$} q[+4];
      qubit {$\ket{+}$} q[+1];
      qubit {$\ket{\makebox[\slength][c]{0}}$} q[+1];
      qubit {$\ket{+}$} q[+2];
      qubit {$\ket{\makebox[\slength][c]{0}}$} q[+10];
cnot q[30] | q[9];
cnot q[22] | q[5];
cnot q[23] | q[12];
cnot q[27] | q[4];
cnot q[24] | q[3];
cnot q[7] | q[6];
cnot q[29] | q[0];
barrier q;
cnot q[16] | q[5];
cnot q[14] | q[9];
cnot q[29] | q[22];
cnot q[26] | q[24];
cnot q[25] | q[19];
cnot q[27] | q[10];
cnot q[23] | q[8];
cnot q[21] | q[3];
cnot q[4] | q[1];
cnot q[18] | q[20];
barrier q;
cnot q[12] | q[16];
cnot q[15] | q[9];
cnot q[14] | q[6];
cnot q[5] | q[2];
cnot q[28] | q[25];
cnot q[22] | q[19];
cnot q[3] | q[11];
cnot q[13] | q[10];
cnot q[18] | q[8];
cnot q[0] | q[21];
barrier q;
cnot q[16] | q[26];
cnot q[14] | q[24];
cnot q[9] | q[20];
cnot q[2] | q[17];
cnot q[1] | q[15];
cnot q[13] | q[11];
barrier q;
cnot q[16] | q[14];
cnot q[19] | q[17];
cnot q[6] | q[29];
cnot q[21] | q[27];
cnot q[30] | q[23];
cnot q[15] | q[12];
cnot q[8] | q[9];
cnot q[28] | q[5];
cnot q[26] | q[4];
cnot q[10] | q[3];
cnot q[25] | q[2];
cnot q[24] | q[1];
cnot q[7] | q[0];
  \end{yquant}
\end{tikzpicture}}
\caption{Non-fault-tolerant $\ket{0}_L$ of the $\llbracket 31, 1, 7 \rrbracket$ color code.}
  \label{fig:d7-color}
\end{sidewaysfigure*}

\end{document}